\documentclass[a4paper,11pt]{article}
\usepackage{graphicx} 
\usepackage[T1]{fontenc}
\usepackage{lmodern}

\usepackage{fullpage}

\usepackage[utf8]{inputenc} 
\usepackage[T1]{fontenc}    
\usepackage{hyperref}       
\usepackage{url}            
\usepackage{booktabs}       
\usepackage{amsfonts}       
\usepackage{nicefrac}       
\usepackage{microtype}      
\usepackage{lipsum}		
\usepackage{graphicx}
\usepackage{natbib}
\usepackage{doi}
\usepackage{mathtools}
\usepackage{amsthm}
\usepackage{amssymb}
\usepackage{tcolorbox}

\usepackage[ruled]{algorithm2e} 
\usepackage{bbm}

\newtheorem{example}{Example}[section]

\newtheorem{definition}{Definition}[section]

\newtheorem{theorem}{Theorem}[section]
\newtheorem{lemma}[theorem]{Lemma}
\newtheorem{claim}[theorem]{Claim}

\newcommand{\be}{\begin{equation}}
\newcommand{\ee}{\end{equation}}
\newcommand{\beq}{\begin{equation*}}
\newcommand{\eeq}{\end{equation*}}

\newcommand{\R}{\mathbb{R}}

\newcommand{\eps}{\varepsilon}

\newcommand{\AutoAdjust}[3]{\mathchoice{ \left #1 #2  \right #3}{#1 #2 #3}{#1 #2 #3}{#1 #2 #3} }
\newcommand{\Xcomment}[1]{{}}

\newcommand{\InBrackets}[1]{\AutoAdjust{[}{#1}{]}}
\newcommand{\Ex}[2][]{\operatorname{\mathbf E}_{#1}\InBrackets{#2}}
\newcommand{\Exlong}[2][]{\operatornamewithlimits{\mathbf E}\limits_{#1}\InBrackets{#2}}
\newcommand{\Prx}[2][]{\operatorname{\mathbf{Pr}}_{#1}\InBrackets{#2}}
\newcommand{\Prlong}[2][]{\operatornamewithlimits{\mathbf{Pr}}\limits_{#1}\InBrackets{#2}}

\newcommand{\eqdef}{\overset{\mathrm{def}}{=\mathrel{\mkern-3mu}=}}
\newcommand{\vect}[1]{\ensuremath{\mathbf{#1}}}

\newcommand{\RN}[1]{
  \textup{\uppercase\expandafter{\romannumeral#1}}
}

\newcommand\restr[2]{{
  \left.\kern-\nulldelimiterspace 
  #1 
  \vphantom{\big|} 
  \right|_{#2} 
  }}
\def\prob{\Prx}

\newcommand{\alg}{\textsf{ALG}}
\newcommand{\opt}{\textsf{OPT}}

\newcommand{\fpi}[2]{#1^{#2}}

\newcommand{\vals}{\vec{v}}

\newcommand{\dist}{\mathbf{F}}
\newcommand{\dists}{\vect{\dist}}

\title{Choosing Behind the Veil:\\ Tight Bounds for Identity-Blind Online Algorithms\thanks{
The work of M. Feldman has been partially funded by the European Research Council (ERC) under the European Union's Horizon 2020 research and innovation program (grant agreement No. 866132), by an Amazon Research Award, by the NSF-BSF (grant No. 2020788), and by a grant from TAU Center for AI and Data Science (TAD). T.~Ezra is supported by the Harvard University Center of Mathematical Sciences and Applications. Z.~Tang is supported by Program for Innovative Research Team of Shanghai University of Finance and Economics (IRTSHUFE) and the Fundamental Research Funds for the Central Universities.
}}
\author{
Tomer Ezra\thanks{Harvard University. Email: \texttt{tomer@cmsa.fas.harvard.edu}}
\and
Michal Feldman\thanks{Tel Aviv University, Email: \texttt{michal.feldman@cs.tau.ac.il }}
\and
Zhihao Gavin Tang\thanks{ITCS, Key Laboratory of Interdisciplinary Research of Computation and Economics, Shanghai University of Finance and Economics. Email: \texttt{tang.zhihao@mail.shufe.edu.cn}}
}
\date{}
\begin{document}

\begin{titlepage}
\maketitle
\thispagestyle{empty}

\begin{abstract}

In Bayesian online settings, every element has a value that is drawn from a known underlying distribution, which we refer to as the element's {\em identity}. 
The elements arrive sequentially.
Upon the arrival of an element, its value is revealed, and the decision maker needs to decide, immediately and irrevocably, whether to accept it or not.
While most previous work has assumed that the decision maker, upon observing the element's value, also becomes aware of its identity --- namely, its distribution --- practical scenarios frequently demand that decisions be made based solely on the element's value, without considering its identity.
This necessity arises either from the algorithm's ignorance of the element's identity or due to the pursuit of fairness, ensuring uniform decisions across different identities. 
We call such algorithms {\em identity-blind} algorithms, and propose the {\em identity-blindness gap} as a metric to evaluate the performance loss in online algorithms caused by identity-blindness. This gap is defined as the maximum ratio between the expected performance of an identity-blind online algorithm and an optimal online algorithm that knows the arrival order, thus also the identities.

We study the identity-blindness gap in the paradigmatic prophet inequality problem, under the two common objectives of maximizing the expected value, and maximizing the probability to obtain the highest value, and provide tight bounds with respect to both objectives. 
For the max-expectation objective, the celebrated prophet inequality establishes a single-threshold (thus identity-blind) algorithm  that gives at least $1/2$ of the offline optimum, thus also an identity-blindness gap of at least $1/2$. 
We show that this bound is tight even with respect to the identity-blindness gap.
We next consider the max-probability objective. While the competitive ratio is tightly $1/e$, we provide a deterministic single-threshold (thus identity-blind) algorithm that gives an identity-blindness gap of $\sim 0.562$ under the assumption that there are no large point masses. Moreover, we show that this bound is tight with respect to deterministic algorithms. 
\end{abstract}
\end{titlepage}

\section{Introduction}

In online settings, the input arrives sequentially over time, and the {\em online algorithm} makes decisions online, without knowing future arrivals.
A paradigmatic online problem in Bayesian settings is the {\em prophet inequality} \citep{krengel1977semiamarts,krengel1978semiamarts,samuel1984comparison}. 
In a prophet inequality setting, there are $n$ boxes, arriving online, each box $t=1,\ldots,n$ contains a value $v_t$ drawn from a known probability distribution $F_t$. 
Upon the arrival of box $t$, its value $v_t$ is revealed, and the online algorithm needs to decide immediately and irrevocably whether to accept it or not. 
If accepted, the game ends with a reward of $v_t$. 
Otherwise, the reward $v_t$ is lost forever and the game proceeds to the next box (if any).
The prophet inequality problem captures many real-life scenarios, and in particular market design problems and the design of pricing mechanisms \citep{hajiaghayi2007automated}. 

In the analysis of online algorithms, the primary metric used to assess performance is known as the {\em competitive ratio}, defined as the worst-case (across all distributions) ratio between the expected performance of the online algorithm and that of the optimal offline algorithm. Essentially, the performance of an online algorithm is gauged by comparing its expected performance to that of a hypothetical ``prophet", capable of making decisions with future knowledge. 

It is well-known that the prophet inequality problem admits an online algorithm that has a competitive ratio of a half. Remarkably, the optimal competitive ratio can be achieved by a simple single threshold algorithm that accepts the first reward that exceeds the threshold, while no online algorithm can obtain a better competitive ratio. 

However, the prophet benchmark is quite strong and hypothetical, and the performance of online algorithms has also been measured with respect to weaker benchmarks. 
One such benchmark is the optimal online algorithm. 
Research in this direction focuses on comparing the best polynomial-time online algorithm to the best online algorithm \citep{papadimitriou2021online,braverman2022max,srinivasan2023online}, quantifying the loss in online algorithms due to computational constraints. 
Similarly, \cite{niazadeh2018prophet} study the ratio between the best single-threshold algorithm and the best overall online algorithm, quantifying the loss that arises due to the simplicity inherent in single-threshold algorithms. 

Building upon this line of research, \cite{ezra23order} recently introduced the {\em order-competitive ratio} measure, quantifying the loss in online algorithms that arises due to the lack of knowledge about the arrival order of elements. The order-competitive ratio is defined as the ratio of the performance of the best online algorithm that operates without knowledge of the arrival order, to that of the best online algorithm that has full knowledge of the arrival order.

A fundamental assumption central to the analysis in \cite{ezra2018prophets} concerns the information available to the algorithm upon the arrival of each element. Specifically, even though the algorithm does not know the arrival order, it is assumed to be informed of each element's {\em identity} upon its arrival. Namely, upon the arrival of an element at time $t$, with a revealed value $v_t$, the algorithm is aware of the underlying distribution from which $v_t$ has been drawn and can use this information (along with the identities of elements that arrived earlier) in its decision whether to accept it or not.

In many practical scenarios, however, decision-making must rely solely on the value of an element, without accounting for its identity (i.e., population). This requirement is often necessitated by either the algorithm's lack of knowledge about the element's identity or a deliberate choice aimed at ensuring fairness, thus guaranteeing uniform decisions for different identities.

A prime example of this can be seen in item sales, a topic extensively explored within the prophet inequality framework \citep{hajiaghayi2007automated,DBLP:conf/stoc/ChawlaHMS10}.
In such contexts, there is a strong preference for non-discriminatory pricing strategies. These strategies offer the same price to customers presenting identical values, regardless of the customer's specific identity, which is typically represented by the underlying distribution of their population.
This approach addresses fairness concerns and regulatory compliance standards. Importantly, this pricing strategy may be preferred even in situations where sellers are aware of the buyers' identities.

As another illustrative example, orchestras globally conduct auditions behind a curtain, ensuring judges cannot see the performer's identity. This practice of blind auditions has been established as a norm in symphony orchestras to guarantee that evaluations are based exclusively on the candidate's performance quality (i.e., value) rather than their gender or ethnic background (i.e., underlying distribution).
Research has demonstrated that blind auditions effectively reduce gender biases during the selection process \cite{ClaudiaC00}.
Similarly, in a double-blind review process, the evaluation of submissions is conducted purely based on the academic merit of the work and its contribution to the field (i.e., value), independent of any knowledge of the authors' scholarly reputations or affiliations (i.e., their underlying distribution). 

\paragraph{Identity-blind online algorithms.}
To address this issue, we study the performance of online algorithms that do not know the arrival order, neither the identity of the arriving elements. We call such algorithms {\em identity-blind} algorithms. 
An identity-blind algorithm observes $v_t$ (i.e., the realized value at time $t$), but not $F_t$ (i.e., the underlying distribution from which $v_t$ has been drawn).

We introduce a new measure, called the {\em identity-blindness gap}, defined as the worst-case ratio (over all distribution sequences) between 
the performance of the best {\em identity-blind} algorithm and the best online  algorithm that knows the arrival order, thus also the identities of the arriving elements.
Clearly, the identity-blindness gap is upper bounded by the order-competitive ratio studied in \cite{ezra23order}.

For illustration, consider the following example.

\begin{example} \label{ex:motivation} (see Figure~\ref{fig:example})
Suppose there are three boxes. Box 1 with deterministic value 0, Box 2 with deterministic value $1/2$, and Box 3 with value $1/\epsilon$ with probability $\epsilon$ (and $0$ otherwise), for an arbitrarily small $\epsilon$. 
Suppose further that the objective is to maximize the expected value of the accepted box.   

If the boxes arrive according to the order $1,2,3$, then the optimal online algorithm (that knows the order) rejects the first two boxes, for an expected value of $1$. 
If the boxes arrive according to the order $3,2,1$, then the optimal online algorithm accepts $1/\epsilon$ if realized in the first box, otherwise rejects it and accept $1/2$ in Box $2$, for an  expected value of $1.5$.

With probability $1-\epsilon$, the first two values in both orders are $0$ followed by $1/2$. 
Thus, unless the ``$\epsilon$ event" happens, an identity-blind algorithm $\alg$ cannot distinguish between the two orders. 
If $\alg$ accepts Box 2 after observing $0,1/2$, then in the first order it gets $1/2$, compared to the optimal value of $1$. 
If $\alg$ declines Box 2 after observing $0,1/2$, then in the second order it gets $1$, compared to the optimal value of $1.5$.
Thus, no deterministic algorithm can guarantee an identity-blindness gap better than $2/3$. 
Note that if we change the value of Box 2 from $1/2$ to $1/\phi$ (where $\phi$ is the golden ratio), then the same example shows that no deterministic identity-blind algorithm can guarantee a gap better than $1/\phi \approx 0.618$.
\label{ex:gap}
\end{example}

\begin{figure}
\centering
\includegraphics[width=0.9\textwidth]
{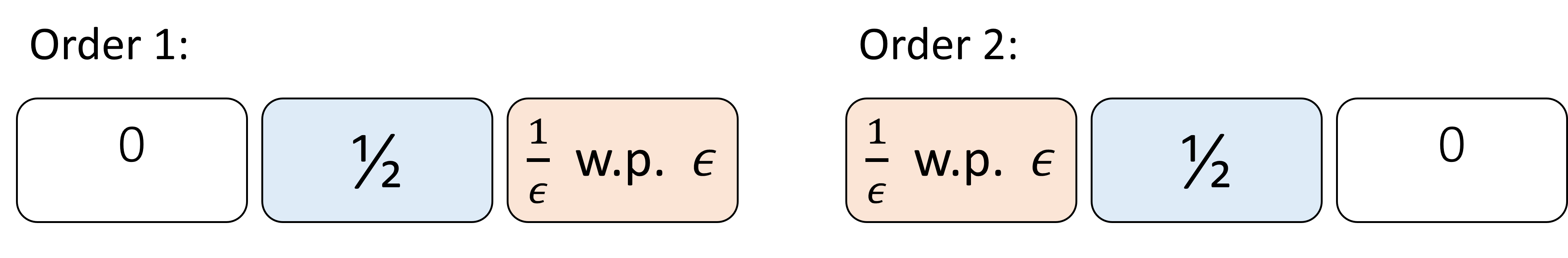}
\caption{An example showing an upper bound of $2/3$ on the identity-blindness gap of deterministic algorithms.
\label{fig:example}}
\end{figure}

\subsection{Our Results and Techniques}

We provide tight bounds for the identity-blindness gap, with respect to the following two extensively-studied  objective functions: 
\begin{itemize}
    \item Max expectation objective: maximizing the expected accepted value.
    \item Max probability objective: maximizing the probability to accept the maximum value.
\end{itemize}

\paragraph{{\bf Max expectation objective.}}
The celebrated prophet inequality establishes the existence of a single-threshold algorithm that achieves a competitive ratio of $1/2$ (with respect to the prophet benchmark). 
Since a single-threshold algorithm is identity-blind by definition, and the prophet benchmark is stronger than the best online algorithm benchmark, this result gives us an immediate lower bound of $1/2$ on the identity-blindness gap.

Our main result for the max expectation objective is that $1/2$ is tight with respect to the identity-blindness gap, and even with respect to randomized strategies. 

\vspace{0.1in}
\noindent {\bf Main Theorem 1:} No identity-blind algorithm, deterministic or randomized, can achieve a better identity-blindness gap than $1/2$ with respect to the max expectation objective. 
\vspace{0.1in}

Notably, this result establishes a separation between the performance of (order-unaware) algorithms that can observe the identities of the arriving elements and those that cannot. Indeed, \cite{ezra23order} establish a superior order-competitive ratio of $0.618$ (the inverse of the golden ratio) using an algorithm that takes the identities into account. 
This result thus present a separation between what can be achieved with and without discrimination based on identities.

Our proof relies on extending the original hardness example for the prophet setting (where there is a  deterministic box with a value of $1$, followed by a randomized box with a value of $1/\varepsilon$ with probability $\varepsilon$ and $0$ otherwise) to an instance in which an identity-blind algorithm cannot distinguish which of the boxes arrived already or not.
We do so by creating an instance with four types of boxes: boxes with a value of $0$, Bernoulli boxes with a value $1$, a deterministic box with a value of $1$, and a randomized box with a value of $1/\varepsilon$ with probability $\varepsilon$, and $0$ otherwise.
We then construct a distribution over arrival orders, where the Bernoulli boxes arrive first (mixed with some of the $0$ value boxes) followed by the randomized box, followed by the value $1$ box, followed by the remaining value $0$ boxes. An algorithm that knows that this is the arrival order, can get a value of $2-\varepsilon$ by selecting the randomized box if realized, and if not selecting the deterministic box with value $1$ (that always arrives after).
An identity-blind algorithm, cannot tell from only observing the values whether all the Bernoulli boxes arrive already or not, and cannot guess it with high enough probability to obtain a value that is significantly larger than $1$.

\paragraph{{\bf Max probability objective.}}
\citet{esfandiari2017prophet} devise an order-unaware single-threshold algorithm that gives a tight (up to lower-order terms) competitive ratio of $1/e$ (with respect to the prophet benchmark). 
Since an order-unaware single-threshold algorithm is by definition identity-blind, and since the prophet benchmark is at least as high as the best online algorithm benchmark, the $1/e$ ratio applies to the identity-blindness gap as well.
On the other hand, the upper bound of $0.806$ on the order-competitive ratio of \citet{ezra23order} applies here as well. Thus, the identity-blindness gap lies between $1/e$ and $0.806$.

Our second main result is a
tight (with respect to deterministic algorithms) bound on the identity-blindness gap with respect to the max probability objective. 

\vspace{0.1in}
\noindent {\bf Main Theorem 2:} There exists a deterministic single-threshold (thus identity-blind) algorithm that gives an identity-blindness gap of at least $\sim 0.562$ (see remark below). 
Moreover, no deterministic algorithm can obtain a better identity-blindness gap.
\vspace{0.1in}

\noindent {\bf Remark:} The positive result holds under a natural assumption that the distribution of the maximum value has no large point masses.
In particular, this assumption implies the existence of a unique maximum with high probability.
The identity-blindness gap approaches $0.562$ as the bound on the point masses goes to $0$.
A discussion on relaxations of this assumption is given in Appendix~\ref{sec:discussion}.

Although a single-threshold algorithm can be captured by a single number (i.e., the threshold itself or the corresponding probability that the maximum value is below it), the optimal algorithm can be adaptive and complicated to describe. This is in stark contrast to the single-threshold versus single-threshold analysis of \citet{ezra23order}, which is essentially a two-parameter optimization problem. To this end, we start by guessing out the worst-case instance for an arbitrary single-threshold algorithm. In principle, there exists a trade-off so that the threshold should neither be too high, nor too low. We restrict ourselves to Bernoulli instances and find that the worst-case instance against an arbitrary single-threshold algorithm involves three stages. Guided by this family of instances, we establish our positive result for general value distributions, where we introduce multiple variants of the game to facilitate the analysis. This is the most technical part of our analysis.

{\bf Discussion.}
We believe that the fact that the best algorithms in our scenarios (for both objectives) are single threshold algorithms is quite remarkable and merits further discussion. 
In the classic prophet inequality (with the prophet as a benchmark), the best competitive ratios (for both objectives) are obtained by single-threshold algorithms \citep{samuel1984comparison,esfandiariHLM20}.
In contrast, \cite{ezra23order} showed that under the benchmark of the best online algorithm, in scenarios where identities are observed, superior ratios are achieved by adaptive algorithms, not attainable through single-threshold algorithms.
Intriguingly, in our setting where identities are unobserved and using the best online algorithm as a benchmark, we see a resurgence of the classic results; specifically, the best ratios are again achieved by static threshold algorithms. 
Moreover, for the max-probability objective, the fixed-threshold algorithm that delivers this result differs from the one employed in the original prophet inequality problem.

\subsection{Related Work}
\paragraph{Different arrival orders.} Several studies have explored various assumptions regarding the order of arrivals, in addition to the adversarial order considered by \citet{krengel1977semiamarts,krengel1978semiamarts,samuel1984comparison}. These studies have considered different scenarios, such as random arrival order (also known as the prophet secretary problem) \citep{esfandiari2017prophet,azar2018prophet,ehsani2018prophet,CorreaSZ21}, as well as free-order settings where the algorithm can dictate the arrival order, as studied by \cite{beyhaghi2018improved,AgrawalSZ20,PengT22,BubnaC23}.
Furthermore, a recent study by \citet{arsenis2021constrained} 
has demonstrated that regardless of the specific arrival order $\pi$, the better option between $\pi$ and the reverse order of $\pi$ achieves a competitive ratio of at least the inverse of the golden ratio.

\paragraph{Alternative benchmarks.}
Recently, there has been
a significant interest in exploring alternative benchmarks to the ``prophet'' benchmark, which represents the optimal offline solution. This interest is reflected in the studies of \citet{kessel2021stationary,niazadeh2018prophet,papadimitriou2021online,ezra23order,DBLP:conf/sigecom/DuttingGRTT23,DBLP:conf/wine/EzraG23}.
For instance, \citet{niazadeh2018prophet} focus on quantifying the loss incurred by single-threshold algorithms compared to the best adaptive online algorithm (which can be single-threshold or not) under a known order. They establish that the tight worst-case ratio between these algorithms is $1/2$.
Another example is presented by \citet{papadimitriou2021online}, who examine the problem of online matching in bipartite graphs. While a 1/2-competitive algorithm exists for this problem concerning the prophet benchmark \citep{FeldmanGL15}, \citet{papadimitriou2021online} introduce a different perspective. They consider the ratio between the optimal polynomial algorithm (which knows the arrival order) and the optimal computationally-unconstrained algorithm (which also knows the arrival order), showcasing that this ratio exceeds $1/2$. (It's worth noting that in the matching variant, finding the optimal online algorithm for matching is computationally challenging even when the order is known.) Subsequent work by \citet{SaberiW21} has improved this bound to $0.526$, which has been improved further to to $1 - 1/e$ \citep{braverman2022max}, and yet further to $0.652$ \citep{srinivasan2023online}.

\paragraph{Beyond single-choice.}
In addition to single-choice settings, there is a related body of research that expands upon the optimal stopping problem to encompass multiple-choice settings. This line of work was initiated by \citet{kennedy1985optimal,kennedy1987prophet,kertz1986comparison}. 
Furthermore, recent advancements have extended this framework to various combinatorial scenarios, including matroids \citep{kleinberg2012matroid,azar2014prophet}, polymatroids \citep{dutting2015polymatroid}, matching markets \citep{GravinW19,EzraFGT20}, combinatorial auctions \citep{FeldmanGL15,dutting2020log}, and general downward-closed (and beyond) feasibility constraints \citep{rubinstein2016beyond}.

\paragraph{Fairness in Prophet inequality.}
Our research is related to a new line of work that incorporates fairness considerations into prophet settings \citep{50313,10.1145/3490486.3538301}.
\citet{50313} study prophet and secretary variants in which the decision-maker is restricted to select boxes proportionally according to predefined ratios (and so is the prophet). \citet{10.1145/3490486.3538301} suggest two fairness notions for prophet settings, namely, identity-independent fairness and time-independent fairness. The main difference in the approach of \citet{50313,10.1145/3490486.3538301} and ours, is that they try to define conditions which they view as fair, and study the implications of imposing these conditions, while we study the implications of hiding the discriminatory information from our algorithms.

\section{Model and Preliminaries}
We consider a setting with $n$ boxes. 
Every box $i$ contains some value $v_i$ drawn from an underlying independent distribution $F_i$.
The distributions $F_1, \ldots, F_n$ are known from the outset. 
The values are revealed sequentially in an online fashion, according to an unknown order $\pi$, i.e., $\pi(i)$ denotes the identity of the $i$-th arriving box. We denote by $\fpi{\dists}{\pi}$ the permuted product distribution $F_{\pi(1)} \times \ldots \times F_{\pi(n)}$.

\paragraph{Identity-blind algorithms.}
An online algorithm is said to be {\em identity-blind} if, upon the arrival of box $\pi(i)$, the algorithm observes the revealed value of box $\pi(i)$, 
but is unaware of the box from which the value has been drawn (namely, the index $\pi(i)$ of the arriving box).
Upon observing the revealed value, the algorithm needs to decide, immediately and irrevocably, whether to accept it or not.

To illustrate the setting, in Example~\ref{ex:motivation}, $n=3$, and $F_1$ is always $0$, $F_2$ is always $1/2$, and $F_3$, is $1/\varepsilon$ with probability $\varepsilon$, and $0$ otherwise.   Under both $\pi^1=(1,2,3)$ and $\pi^2=(3,2,1)$, with probability of at least $1-\varepsilon$, the first two values that the algorithm observes are $0,1/2$, so an identity-blind algorithm must treat this scenario the same, while an order-aware algorithm (that knows $\pi$) can treat these two scenarios differently.

Our goal is to measure the performance of identity-blind algorithms against the benchmark of the best online order-aware algorithm.
We call this ratio the {\em identity-blindness gap}.
This measure is related to the {\em order-competitive ratio} of \cite{ezra23order}, which also measures the ratio between the performance of order-unaware and order-aware algorithms. 
However, while the order-competitive ratio measures the performance of order-unaware algorithms that know the identity of the arriving box, we consider identity-blind algorithms that are not even aware of the box identity. 

Recall that we consider two different objectives, namely the max expectation objective and the max probability objective.

\paragraph{Notation and formal model.}
An algorithm will be denoted by $\alg$. 
Given a product distribution $\dists$, 
and a value profile $\vals$, we denote by $\alg(\dists,\vals)$ the value accepted by $\alg$ under the realization of values $\vals$.
(Note that the algorithm may be either order-aware or identity-blind.)

For the max expectation objective, the performance of $\alg$ is the expected accepted value, i.e., 
$
\alg(\dists,\pi)=\Ex[\vals\sim\dists^\pi]{\alg(\dists,\vals)}
$.

For the max probability objective, it is the probability of accepting the maximum value, i.e., 
$
\alg(\dists,\pi)=\prob[\vals\sim\dists^\pi]{\alg(\dists,\vals)=\max_i v_i}.
$
When studying the max probability objective, we assume that the distribution of the maximum value does not have large mass points\footnote{Formally, we assume that there exists $\varepsilon>0$ such that for every $v\in \R_{\geq 0}$, it holds that $\Pr[\max_i v_i = v] \leq \varepsilon$.}, which captures many settings, such continuous distributions,  or distributions that are closed to be continuous. In Section~\ref{sec:discussion} we elaborate on this assumption, and provide stronger impossibility results when the assumption fails to hold.

The identity-blindness gap measures the loss in performance due to unknown order in cases where the identity is not revealed either. 
It is defined as the worst-case ratio of the performance of $\alg$ and the performance of the optimal algorithm $\opt$, over all arrival orders. 
\begin{definition}\label{def:ocr}
The identity-blindness gap of an identity-blind algorithm $\alg$ is 
$$
\Gamma(\alg) = \inf_{\dists} \min_{\pi}\frac{\alg(\dists,\pi)}{\opt_\pi(\dists,\pi)},
$$
where $\opt_\pi$ is the optimal order-aware algorithm for arrival order $\pi$.\footnote{We note that $\opt_\pi(\dists,\pi)$ is well defined, as it can be solved by dynamic programming.}
\end{definition}

We also define the identity-blindness gap of a family of algorithms $\mathcal{A}$ to be the worst-case identity-blindness gap of any algorithm $\alg \in \mathcal{A}$.

\section{Max Expectation Objective}
\label{sec:max_exp}
Our result for the max-expectation objective is negative. We prove that in the worst case, any (randomized) identity-blind online algorithm cannot achieve an identity-blindness gap better than $\frac{1}{2}$. Recall that a simple single-threshold algorithm achieves a competitive ratio of $\frac{1}{2}$ (against the prophet). We conclude that single-threshold algorithms are optimal among identity-blind algorithms.

\begin{theorem}
For any $\varepsilon>0$, there is no (deterministic or randomized) identity-blind algorithm that can achieve a $\left( \frac{1}{2}+\varepsilon \right)$-identity-blindness gap.
\end{theorem}
\begin{proof}
Consider an instance with $2n+2$ boxes of the following $4$ types:
\begin{enumerate}
    \item $F_1,\ldots,F_n$ are deterministic with a value $0$; 
    \item $F_{n+1},\ldots,F_{2n} $ are Bernoulli, whose value is $1$ with probability $\frac{1}{2}$ and $0$ otherwise;
    \item $F_{2n+1}$ is a ``free-reward'' box, whose value is $\frac{1}{\varepsilon}$ with probability $\varepsilon$ and $0$ otherwise;
    \item $F_{2n+2}$ is deterministic with a value $1$.
\end{enumerate}
We consider a family of arrival orders $\pi_{\vect{x}}$ that are parameterized by $\vect{x}\in \{0,1\}^n$.
In the following, we denote $|\vect{x}| \eqdef \sum_{i} x_i$.
\begin{itemize}
    \item For $i=1,\ldots,n$, $\pi_{\vect{x}}(i)$ is of Type~2~\footnote{We note that the specific matching among boxes of the same type does not change the analysis.} if $x_i=1$ and is of Type~1 if $x_i=0$;
    \item For $i=n+1,\ldots, 2n- |\vect{x}| $,  $\pi_{\vect{x}}(i)$ is of Type 2;
    \item  $\pi_{\vect{x}}(2n+1-|\vect{x}|) = 2n+1$ (i.e., Type 3);
    \item  $\pi_{\vect{x}}(2n+2-|\vect{x}|) = 2n+2$ (i.e., Type 4); 
    \item For $i=2n+3-|\vect{x}|,\ldots, 2n+2$,  $\pi_{\vect{x}}(i)$ is of Type 1.
\end{itemize}
Now, consider a random order $\pi_{\vect{X}}$ where $\vect{X}$ is distributed uniformly on $\{0,1\}^n$. We will show that any (deterministic or randomized) identity-blind algorithm has an expected accepted value of $1+ o(1)$, while for each $\pi_{\vect{x}}$ in the support of $\pi_{\vect{X}}$, the expected accepted value of the best order-aware algorithm is $2-o(1)$. 

We first study the performance of the best order-aware algorithm and prove the second statement. Given $\vect{X}=\vect{x}$ (which determines the arrival order), the algorithm that rejects the first $2n-|\vect{x}|$ boxes, then accepts Box $2n-|\vect{x}|+1$ if it is $\frac{1}{\varepsilon}$, otherwise accept Box $2n-|\vect{x}|+2$ (whose value is $1$), achieves an expected accepted value of $2-\varepsilon$.  

To show that identity-blind algorithms cannot do better than $1+o(1)$, we characterize the behavior of the best identity-blind algorithm.
\begin{claim} \label{cl:p4}
Among all identity-blind algorithms, there is one with a maximum expected value that satisfies the following properties:
\begin{enumerate}
    \item[(P1)] The algorithm always selects the $1/\varepsilon$ if observed.
    \item[(P2)] The algorithm always rejects values of $0$.
    \item[(P3)] The algorithm rejects all  values of $1$ up to time $n+1$.
    \item[(P4)] The algorithm rejects values of $1$ that come immediately after another values of $1$.
    \item[(P5)] The algorithm is deterministic.
\end{enumerate}
\end{claim}
\begin{proof}
By Yao's principle since the input is stochastic, there is a deterministic identity-blindness algorithm that maximizes the expected value (thus Property (P5) holds).
To show that there exists a deterministic algorithm that satisfies Properties (P1) and (P2), we observe that there is no larger value than $1/\varepsilon$ thus selecting it is always beneficial, and discarding $0$s is lossless.
For Properties (P3) and (P4), we observe that it cannot be the last value of $1$, therfore we can modify an algorithm to discard it, and accept the next value of $1$ while strictly increasing the expected value (since there is a chance that we will observe the $1/\varepsilon$ value).    
\end{proof}
Let $\alpha_t$ be the event that $\alg$ selected Box $t$, and let $\alpha_{<t}$ (respectively $\alpha_{>t}$) the event that $\alg$ selected a box before time $t$ (respectively, after time $t$). 
Let $Z$ be the event that the free reward is realized, and let $Z_t$ be the event that the value at time $t$ is $1/\varepsilon$. 
The value of the algorithm can be written as 
\begin{align}
\Exlong[\vect{X}]{\alg(\dists,\pi_{\vect{X}})} & = \Prlong[\vect{X}]{\alpha_{<2n+1-|\vect{X}|}} + \Prlong[\vect{X}]{\alpha_{2n+1-|\vect{X}|}} \cdot \frac{1}{\varepsilon} + \Prlong[\vect{X}]{\alpha_{>2n+1-|\vect{X}|}}  \nonumber \\ 
& = \Prlong[\vect{X}]{\alpha_{<2n+1-|\vect{X}|}} + \Prlong[\vect{X}]{Z \wedge \neg\alpha_{<2n+1-|\vect{X}|}} \cdot \frac{1}{\varepsilon} + \Prlong[\vect{X}]{\alpha_{2n+2-|\vect{X}|}} \nonumber \\
& = \Prlong[\vect{X}]{\alpha_{<2n+1-|\vect{X}|}} + \Prlong[\vect{X}]{Z} \cdot \Prlong[\vect{X}]{\wedge \neg\alpha_{<2n+1-|\vect{X}|}} \cdot \frac{1}{\varepsilon} + \Prlong[\vect{X}]{\alpha_{2n+2-|\vect{X}|}} \nonumber \\
& = 1 + \Prlong[\vect{X}]{ \alpha_{2n+2-|\vect{X}|}}, \label{eq:prob}
\end{align}
where the first equality holds since either $\alg$ selects a value of $1$ before the free reward box, or it selects the free reward box if it is realized, or it selects a box with a value of $1$ after the free reward box. The second equality holds since $\alg$ selects the free reward box if the box is realized and $\alg$ does not stop before, and since that the only box with a value of $1$ after the free reward box is Box $2n+2-|\vect{X}|$. The third equality holds since the events $Z$ and $\neg\alpha_{<2n+1-|\vect{X}|}$ are independent. The last equality holds since $\Prx{Z} = \varepsilon$.

Our goal is then to show that $\Prx{\alpha_{2n+2-|\vect{X}|}} = o(1)$.
Notice that
\[
\Prx{\alpha_{2n+2-|\vect{X}|}} = \Prx{\alpha_{2n+2-|\vect{X}|} \wedge Z} + \Prx{\alpha_{2n+2-|\vect{X}|} \wedge \neg Z} \le \varepsilon + \Prx{\alpha_{2n+2-|\vect{X}|} \wedge \neg Z}~.
\]
Hereafter, we study the best algorithm that maximizes $\Pr[\alpha_{2n+2-|\vect{X}|} \wedge \neg Z]$.

The rest of the proof is mostly technical. We give a high-level plan before we delve into the details of the analysis. Observe that the game is essentially to guess $|\vect{X}|$ exactly, which is distributed as $\text{Bin}(n,1/2)$. 
Suppose that we need to make the guess right after seeing the first $n$ boxes, then we should only count the number of 1's that we have seen. However, for most values of the number of 1's (which is equivalent to having an additional sample of $\text{Bin}(|\vect{X}|,1/2)$), the random variable $|\vect{X}|$ is (approximately) equally likely to be distributed between $\frac{1}{2} n \pm \Theta(\sqrt{n})$, which makes guessing exactly $|\vect{X}|$ to be difficult. To generalize the argument for guessing at time $t>n$, we need to deal with the extra observed values from time $n+1$ to $t$. Indeed, we argue that such information is not too helpful. 

Let $\vect{Y}$ be the observed value sequence of the instance, where $Y_t$ is the observed value at time $t$, and $\vect{Y}_{\leq t}$ be the observed value sequence up to time $t$.  We show that 
the best algorithm's decision at time $t$ depends only on $|\vect{Y}_{\leq n}|$, $Y_t$, and $Y_{t-1}$.

\begin{claim}\label{cl:generic}
There exists an identity-blind algorithm $\alg$ with maximum expected utility that satisfies Properties (P1-P5) that also satisfies the following property:
\begin{enumerate}
    \item[(P6)] For $t>n+1$, given that $Y_t=1$, and $Y_{t-1}=0$ the decision of $\alg$ whether to accept $t$ depends only on the value of $| Y_{\leq n}|$.
\end{enumerate} 
\end{claim}
\begin{proof}
Let $\alg'$ be an algorithm that satisfies Properties P(1-5) with maximum expected value.  
Now consider an algorithm $\alg$ that always selects $1/\varepsilon$, and at time $t>n+1$, it selects $t$, if $Y_t=1 \wedge Y_{t-1}=0$, and $|\vect{Y}_{\leq n}|=\delta$, with probability\footnote{Where $\frac{0}{0}$ can be interpreted arbitrarily since it means that this state is not reachable.} 
\[
\frac{ \Prx{|\vect{Y}_{\leq n} |=\delta \wedge Y_t=1 \wedge Y_{t-1}=0 \wedge \alpha'_t}}{\Prx{|\vect{Y}_{\leq n} |=\delta \wedge Y_t=1 \wedge Y_{t-1} =0\wedge \neg \alpha'_{<t}}}
,\]
where $\alpha'_t$ and $\alpha'_{<t}$  are the events that $\alg'_t$  selected $t$ or before time $t$, respectively.  
Then, by design, for every $\delta \in \{0,\ldots,n\}$, the probabilities that after observing $\vect{Y}_{\leq n}$ for which $|\vect{Y}_{\leq n}|=\delta $ that $\alg$ and $\alg'$ arrive to time $t$ are the same, and also the probabilities that $\alg$ and $\alg'$  select $t$, given $Y_{t}=1 \wedge Y_{t-1}=0$, and that $1/\varepsilon$ was not observed.

Since the randomness of $\alg$ can be drawn at the beginning (before observing the input), then it is a distribution over deterministic algorithms that each satisfies the Properties (P1-6), therefore, there exists a deterministic algorithm that satisfies the claim.
\end{proof}

The next claim states that the probability of selecting the deterministic box is monotone in $t$.
\begin{claim}\label{cl:calc}
    For every $\delta\in\{0,\ldots,n\}$, and any $t \in \{n+2,\ldots,2n+1-\delta \}$, it holds that 
    \begin{eqnarray*}
    & & \Prx{t=2n+2-|\vect{X}| \mid |\vect{Y}_{\leq n}| =\delta \wedge Y_t=1 \wedge Y_{t-1} =0 \wedge \neg Z_{\leq t}} \\ & \leq &\Prx{t+1=2n+2-|\vect{X}| \mid |\vect{Y}_{\leq n}|= \delta \wedge Y_{t+1}=1 \wedge Y_{t} =0 \wedge \neg Z_{\leq t+1}}
    \end{eqnarray*}
\end{claim}
\begin{proof}
We start by rewriting the conditional probability.
\begin{eqnarray}\label{eq:mono}
    & & \Prx{t=2n+2-|\vect{X}| \mid |\vect{Y}_{\leq n}| =\delta \wedge Y_t=1 \wedge Y_{t-1} =0 \wedge \neg Z_{\leq t} }  \nonumber \\& = & \frac{\Prx{t=2n+2-|\vect{X}| \wedge |\vect{Y}_{\leq n}| =\delta \wedge Y_t=1 \wedge Y_{t-1} =0 \wedge \neg Z_{\leq t} }}{\Prx{ |\vect{Y}_{\leq n}| =\delta \wedge Y_t=1 \wedge Y_{t-1} =0 \wedge \neg Z_{\leq t} }}
\end{eqnarray}
For the numerator, we have that
\begin{eqnarray}\label{eq:2}
& &     \Prx{t=2n+2-|\vect{X}| \wedge |\vect{Y}_{\leq n}| =\delta \wedge Y_t=1 \wedge Y_{t-1} =0 \wedge \neg Z_{\leq t} } \nonumber
\\     & = & \Prx{\text{Bin}(n,1/2) = 2n+2-t} \cdot \Prx{\text{Bin}(2n+2-t,1/2) =\delta} \cdot (1-\varepsilon) \nonumber
\\     & = & \Prx{\text{Bin}(n,1/4) = \delta} \cdot \Prx{\text{Bin}(n-\delta,1/3) =2n+2-t-\delta} \cdot (1-\varepsilon),
\end{eqnarray}
    where the second equality holds by first drawing $|\vect{Y}_{\leq n }|$ and then drawing $|\vect{X}|$ instead of first drawing $|\vect{X}|$ and then $|\vect{Y}_{\leq n }|$.
Next, we study the denominator 
\begin{eqnarray}\label{eq:3}
    & & \Prx{ |\vect{Y}_{\leq n}| =\delta \wedge Y_t=1 \wedge Y_{t-1} =0 \wedge \neg Z_{\leq t} } \nonumber \\
    &= & \Prx{\text{Bin}(n,1/2) = 2n+2-t} \cdot \Prx{\text{Bin}(2n+2-t,1/2) =\delta} \cdot (1-\varepsilon) 
    \nonumber \\ 
    &+ & \sum_{\gamma=\delta}^{2n-t} \Prx{\text{Bin}(n,1/2) = \gamma} \cdot \Prx{\text{Bin}(\gamma,1/2) =\delta} \cdot \frac{1}{4}  \nonumber \\
    & = & \Prx{\text{Bin}(n,1/4) = \delta} \cdot \Prx{\text{Bin}(n-\delta,1/3) =2n+2-t-\delta} \cdot (1-\varepsilon)\nonumber\\
    &+ & \sum_{\gamma=\delta}^{2n-t} \Prx{\text{Bin}(n,1/4) = \delta} \cdot \Prx{\text{Bin}(n-\delta ,1/3) =\gamma-\delta} \cdot \frac{1}{4} 
\end{eqnarray}
Notice that for every $a,b,c,d \in \R_{\geq 0}$, $\frac{a}{a+b}  \leq \frac{c}{c+d}$ if and only if $\frac{a}{b} \leq \frac{c}{d}$. By substituting $a,b$ by equations~\eqref{eq:2} and \eqref{eq:3} respectively, and substituting $c,d$ by the corresponding equations for $t+1$, the claim follows since 
\[
4\cdot \frac{\Prx{\text{Bin}(n-\delta,1/3) =2n+2-t-\delta}  \cdot (1-\varepsilon)}{\sum_{\gamma=\delta}^{2n-t} \Prx{\text{Bin}(n-\delta ,1/3) =\gamma-\delta}}
\]
 is proportional to 
 $$ \frac{2^{t -n -2} {n-\delta \choose 2n+2-t-\delta } }{ \sum_{\gamma =\delta}^{2n-t} 2^{\delta-\gamma} {n-\delta \choose \gamma-\delta }} ,$$
 which by renaming $n' =n-\delta$, $t'=t-n-2$ and $\gamma'=\gamma-\delta$ is 
 $$ \frac{2^{t'} {n'\choose t'} }{ \sum_{\gamma' =0}^{n' -t' -2} 2^{-\gamma'} {n' \choose \gamma' }} $$
 which is weakly\footnote{Here we interpret $c/0$ as infinity.} monotone increasing in $t'$ for every $n'$. 
\end{proof}

We claim that the optimal algorithm is deterministic and monotone in the sense that: 
\begin{claim}\label{cl:monotone}
There exists a mapping $g:\{0,\ldots,n\} \rightarrow \{n+2,\ldots 2n+2\} $, 
such that there exists an identity-blind algorithm $\alg$ with maximum expected utility that satisfies Properties (P1-6), and selects Box $t>n+1 $ with $Y_t=1, Y_{t-1}=0$ if and only if $t\geq g(|\vect{Y}_{\leq n}|)$. 
\end{claim}
\begin{proof}
Let $\delta=|\vect{Y}_{\leq n}|$. By Claim \ref{cl:generic}, we can assume that there is an algorithm $\alg'$ that satisfies Properties (P1-6) with a maximum expected utility.
Assume that $\alg'$  after observing $\vect{Y}_{\leq t}$ for $n+1<t<2n+2$ for which  $|\vect{Y}_{\leq n
}| = \delta \wedge Y_{t}=1 \wedge Y_{t-1} =0$, it selects it, but after observing $\vect{Y}_{\leq t+1}$  for which  $|\vect{Y}_{\leq n
}| = \delta \wedge Y_{t+1}=1 \wedge Y_{t} =0$, it does not select  it.
Let $\beta_{t',\delta}$ be the probability that $\alg'$ reaches time $t'$ for which $|\vect{Y}_{\leq n}| = \delta \wedge Y_{t'}=1 \wedge Y_{t'-1} =0$.

If $\beta_{t,\delta} > \beta_{t+1,\delta}$, then let algorithm $\alg$, that selects with probability $1-\frac{\beta_{t+1,\delta}}{\beta_{t,\delta}}$ if $Y_t=1 \wedge Y_{t-1}=0 \wedge |\vect{Y}_{\leq n}|=\delta$, and with probability $1$  if $Y_{t+1}=1 \wedge Y_{t}=0 \wedge |\vect{Y}_{\leq n}|=\delta$, and when observing $Y_{t+2}=Y_t=1 \wedge Y_{t+1}=Y_{t-1}=0  \wedge |\vect{Y}_{\leq n}|=\delta$, it does not select (in all other cases, it does the same as $\alg'$). 
It is easy to observe, that for each $t'>t+2$, the probabilities that $\alg$ and $\alg'$ reaches time $t'$, where $Y_{t'}=1 \wedge Y_{t'-1}=0 \wedge |\vect{Y}_{\leq n}|=\delta$ are the same, and for $t'=t+2$, the probability that  $\alg'$ reaches a state where $Y_{t'}=1 \wedge Y_{t'-1}=0 \wedge |\vect{Y}_{\leq n}|=\delta$ is the same as the probability that $\alg$ reaches a state where where $Y_{t'}=1 \wedge Y_{t'-1}=0 \wedge |\vect{Y}_{\leq n}|=\delta$, and either $Y_{t'-2}=0$ or $Y_{t'-3}=1$. Thus, by Claim~\ref{cl:calc}, $\alg$ has a better performance than $\alg'$.

Else ($\beta_{t,\delta} \le \beta_{t+1,\delta}$), then the same arguments work for the algorithm  $\alg$ that selects with probability $0$ if $Y_t=1 \wedge Y_{t-1}=0 \wedge |\vect{Y}_{\leq n}|=\delta$, and with probability $1-\frac{\beta_{t,\delta}}{\beta_{t+1,\delta}}$  if $Y_{t+1}=1 \wedge Y_{t}=0 \wedge |\vect{Y}_{\leq n}|=\delta$, and when observing $Y_{t+2}=Y_t=1 \wedge Y_{t+1}=Y_{t-1}=0  \wedge |\vect{Y}_{\leq n}|=\delta$, it does not select (in all other cases, it does the same as $\alg'$).

Since $\alg$ is a distribution over deterministic algorithms, there exists a deterministic algorithm in its support with at least as high expected probability of choosing Box $t=2n+2-|X|$. Note that the deterministic algorithms in the support, either reject both of them, or accept both of them, or reject $t$ and accept $t+1$.
Applying this step repeatedly (in increasing order) converges to an algorithm that satisfies the assumptions of the claim.
\end{proof}

We next bound the expected values of all algorithms satisfying Properties (P1-6) and are defined by a mapping $g$ (as in Claim~\ref{cl:gy}), for a range of realizations of $Y_{\leq n}$.
\begin{claim}\label{cl:gy}
For every $\delta$ for which $n/6 < \delta< n/3$ and any value $g(\delta)$, if $|\vect{Y}_{\leq n}|=\delta $, then for the deterministic algorithm $\alg$ that selects if $1/\varepsilon$ is observed or if for Box $t > n+1 $ only  or if $Y_t=1$, $Y_{t-1}=0$, and $t\geq g(\delta)$ the probability of selecting Box $2n+2-|X|$  is $O\left(\frac{ \log n}{\sqrt{n}}\right)$. 
\end{claim}
\begin{proof}
It holds that 
\begin{eqnarray}\label{eq:prob3}
& &  \Prx{\alg \mbox{ selects Box } 2n+2-|X| \mid |\vect{Y}_{\leq n}| =\delta} \nonumber \\ & = & \sum_{\gamma = g(y)}^{2n+2}  \Prx{\alg \mbox{ selects } \gamma\wedge \gamma = 2n+2-|X| \wedge |\vect{Y}_{\leq n}| =\delta } / \Prx{|\vect{Y}_{\leq n}| =\delta},
\end{eqnarray}
where 
\begin{align} 
& \sum_{\gamma = g(y)}^{2n+2}  \Prx{\alg \mbox{ selects } \gamma\wedge \gamma = 2n+2-|X| \wedge |\vect{Y}_{\leq n}| =\delta } \nonumber \\ 
& = \sum_{\gamma = g(y)}^{2n+2}  \Prx{\text{Bin}(n,1/2) = 2n+2-\gamma} \cdot \Prx{\text{Bin}(2n+2-\gamma,1/2) = \delta} \cdot (1-\varepsilon) \cdot  \frac{\gamma-g(\delta)+1}{2^{\gamma-g(\delta)}} \nonumber \\ 
& = \sum_{\gamma = g(y)}^{2n+2}  \Prx{\text{Bin}(n,1/4) = \delta} \cdot \Prx{\text{Bin}(n-\delta,1/3) = 2n+2-\gamma-\delta} \cdot (1-\varepsilon) \cdot  \frac{\gamma-g(\delta)+1}{2^{\gamma-g(\delta)}}, \label{eq:assign}
\end{align}
where the first equality holds by the law of total probability, and the second equality holds since the probability of $Y_i=1$ is $1/4$, and the probability that $X_i=1$ given that $Y_i=0$ is $1/3$.
Combining Equations~\eqref{eq:prob3} and \eqref{eq:assign} with that $\Prx{|\vect{Y}_{\leq n}| =\delta} =\Prx{\text{Bin}(n,1/4)=\delta}$, we get that:

\begin{multline}\label{eq:prob2}
\Prx{\alg \mbox{ selects Box } 2n+2-|X| \mid |\vect{Y}_{\leq n}| =\delta}  \\ 
= \sum_{\gamma = g(y)}^{2n+2}  \Prx{\text{Bin}(n-\delta,1/3) = 2n+2-\gamma-\delta} \cdot (1-\varepsilon) \cdot  \frac{\gamma-g(\delta)+1}{2^{\gamma-g(\delta)}} \leq  O\left(\frac{ \log n}{\sqrt{n}}\right),
\end{multline}
where the inequality holds since that for the range of $\delta$, $\Prx{\text{Bin}(n-\delta,1/3) = 2n+2-\gamma-\delta} = O\left(\frac{1}{\sqrt{n}}\right)$, and for $\gamma>g(\delta) + 3\log n$, $\frac{\gamma-g(\delta)+1}{2^{\gamma-g(\delta)}} \leq O(\frac{1}{n^2})$.
\end{proof} 

Finally, we conclude the proof of the theorem by combining Claims~
\ref{cl:monotone}, and \ref{cl:gy} with Equation~\eqref{eq:prob} and the fact that $\Prx{n/6<|\vect{Y}_{\leq n}|<n/3}=\Prx{n/6<\text{Bin}(n,1/4)<n/3}=1 -o(1)$.
\end{proof}
 
\section{Max Probability Objective}
For simplicity of presentation, in this section, we assume that the distributions of the maximum value is continuous, rather than not having point masses of more than $\varepsilon$. All of our calculations can be adjusted for the case where  the distributions of the maximum value does not have point masses of more than $\varepsilon$, by losing an error factor of $f(\varepsilon)$, where the error function $f:\R_{>0} \rightarrow \R_{>0}$ approaches $0$ (i.e., $\lim_{\varepsilon\rightarrow 0^+} f(\varepsilon) =0$).
We provide the optimal deterministic single-threshold algorithm with identity-blindness gap of $\Gamma^* \eqdef \frac{\lambda \cdot \ln( 1/\lambda)}{\lambda + 1/e} \approx 0.562$, where $\lambda \approx 0.245$ is the solution to the following equation.
\begin{equation}
\label{eq:prob_ratio}
\min_{\rho \ge \lambda} \frac{\lambda \cdot \ln\left( \frac{1}{\rho} \right) + \rho - \lambda}{\lambda + 1 - \frac{\lambda}{\rho}} = \frac{\lambda \cdot \ln\left( \frac{1}{\lambda} \right)}{\lambda + \frac{1}{e}}
\end{equation}

\paragraph{An identity-blind single-threshold algorithm.} 
Let $\alg$ be the single-threshold algorithm that accepts the first box whose value is at least $\tau$, where $\tau$ satisfies
\[
\Prx{\max_{i \in [n]} v_i < \tau} = \lambda~.
\]

\begin{theorem}
\label{thm:max_prob_single}
    The above single-threshold algorithm has an identity-blindness gap of $\Gamma^* \approx 0.562$.
\end{theorem}

\begin{proof}
Consider the following generalization of the game with two extra components: 1) let there be an extra number $\theta$ given in advance, and 2) the game starts from box $j$. Then, the objective is to maximize the probability of catching the box among $\{j,j+1,\ldots,n\}$ with the largest value that exceeds $\theta$.
Note that the original game corresponds to $\theta=0$ and $j=1$.
We use $\opt_j(\theta)$ to denote the optimal winning probability for this generalization of the game. 

Moreover, we bound the value of $\opt_1(0)$ by the two cases: (1) when the maximum value is less than $\tau$, and (2) the maximum value is at least $\tau$. 
The first term is bounded by $\lambda$, since $\lambda$ is the total probability of the stated event.
The second term is bounded by $\opt_1(\tau)$, according to the definition of $\opt_1(\tau)$. Consequently, it is no larger than $\max_j \opt_j(\tau)$. To sum up, we have that 
\[
\opt = \opt_1(0) \le \Prx{\max_{k} v_k < \tau} + \opt_1(\tau) \le \lambda + \max_{j} \opt_j(\tau) ~.
\]

Let $i = \arg\max_{j} \opt_j(\tau)$ and consider the behavior of the optimal algorithm when the game starts at box $i$ and threshold $\theta=\tau$.
Consider the following events:
\begin{itemize}
\item $T_j$: $j$ is the first box among $\{i,i+1,\cdots, n\}$ whose value passes the threshold $\tau$. I.e., $\max_{i\le k < j} v_k < \tau$ and $v_j \ge \tau$.
\item $A_j$: $T_j$ happens and the optimal algorithm accepts box $j$.
\item $R_j$: $T_j$ happens but the optimal algorithm rejects box $j$.
\end{itemize}

For the ease of our analysis, let $p_i \eqdef \Prx{v_i \ge \tau}$. Then we have that $\prod_{i \in [n]} (1-p_i) = \lambda$ and $\Prx{T_j} = \prod_{k=i}^{j-1}(1-p_k) \cdot p_j$.
Furthermore, let $a_j \eqdef \Prx{A_j \vert T_j}$ and $r_j \eqdef \Prx{R_j \vert T_j}$.

We start with the performance of the optimal algorithm. Without loss of generality, we safely assume that the optimal algorithm only accepts a box if its value is at least $\tau$. Therefore, we split the probability space into $\cup_{j=i}^{n} T_j$ and write the performance of the optimal algorithm according to its decision at step $j$ (i.e., $T_j = A_j \cup R_j$).
\begin{align*}
    \opt_i(\tau) & = \sum_{j=i}^{n} \left( \Prx{A_j} \cdot \Prx{ \left. \max_{k>j} v_k \le v_j \right\vert A_j} + \Prx{R_j} \cdot \Ex{ \left. \opt_{j+1}(v_j) \right\vert R_j} \right) \\
    & \le \sum_{j=i}^{n} \left( \Prx{A_j} \cdot \Prx{ \left. \max_{k>j} v_k \le v_j \right\vert A_j} + \Prx{R_j} \cdot \opt_i(\tau) \right) \\
    & = \sum_{j=i}^{n} \left( \Prx{A_j} \cdot \Prx{ \left. \max_{k>j} v_k \le v_j \right\vert A_j} \right) + \left( \sum_{j=i}^{n} \Prx{ R_j} \right) \cdot \opt_{i}(\tau), \tag{$\star$} \label{eq:opt}
\end{align*}
where the inequality holds since when $R_j$ happens, $\opt_{j+1}(v_j) \le \opt_{j+1}(\tau)$. Moreover, $ \opt_{j+1} \le \opt_i(\tau)$ according to the definition of $i$.

Next, we consider the behavior of the single-threshold algorithm. 1) If it stops at box $j$ with $j<i$, it wins if this is the only box whose value exceeds $\tau$; 2) Otherwise, if all boxes before $i$ have value smaller than $\tau$, we compare the behavior of the single-threshold algorithm with $\opt_i(\tau)$ afterwards. Suppose box $j\ge i$ is the first box whose value $v_j \ge \tau$. a) If $\opt_i(\tau)$ accepts (i.e., $A_j$), the two algorithms would have the same winning probability; b) Otherwise $\opt_i(\tau)$ rejects (i.e. $R_j$), the single threshold algorithm still wins if all future boxes have value smaller than $\tau$.
To sum up, we have the following:
\begin{align}
    \alg & \ge \sum_{j=1}^{i-1} p_j \cdot \prod_{k\ne j} (1-p_k) \tag{i} \label{eq:case1}\\
    & + \prod_{k<i} (1-p_k) \cdot \sum_{j=i}^{n} \left( \Prx{A_j} \cdot \Prx{ \left. \max_{k>j} v_k \le v_j \right\vert A_j} \right) \tag{ii} \label{eq:case2} \\
    & + \prod_{k < i} (1-p_k) \cdot \sum_{j=i}^{n} \left( \Prx{R_j} \cdot \prod_{k>j} (1-p_k) \right)  \tag{iii} \label{eq:case3}
\end{align}

Consider the following parameter that plays a central role in our analysis. Intuitively, it is the probability that the single-threshold algorithm behaves the same as the optimal algorithm.
\[
\rho \eqdef \prod_{k<i} (1-p_k) \cdot \left( 1 - \sum_{j=i}^{n} \Prx{ R_j} \right)~.
\] 
Next, we provide upper and lower bounds of \eqref{eq:case1}, \eqref{eq:case2}, and \eqref{eq:case3}.
\begin{lemma}
\label{lem:prob_alg}
We have the following:
\begin{enumerate}
\item $\alg \ge \lambda \cdot \ln \left( \frac{1}{\lambda} \right)$~;
\item $\eqref{eq:case1} + \eqref{eq:case3} \ge \lambda \cdot \ln \left( \frac{1}{\rho}\right)$~;
\item $\eqref{eq:case2} \le \rho - \lambda$~.
\end{enumerate}
\end{lemma}
\begin{proof}
We start with the first statement of the lemma, that is proved by \cite{esfandiariHLM20}. For completeness, we include a proof here.
\begin{multline*}
\alg \ge \sum_{j=1}^{n} \Prx{j \text{ is the unique box with value at least } \tau} 
= \sum_{j=1}^{n} \left(p_j \cdot \prod_{k \ne j}(1-p_k) \right) \\
= \prod_{k=1}^{n}(1-p_k) \cdot \sum_{j=1}^{n} \frac{p_j}{1-p_j} \ge \prod_{k=1}^{n}(1-p_k) \cdot \sum_{j=1}^{n} \ln \left(\frac{1}{1-p_j}\right) = \lambda \cdot \ln \left( \frac{1}{\lambda} \right)~,
\end{multline*}
where the second inequality follows from the fact that $\frac{x}{1-x} \ge \ln\left( \frac{1}{1-x} \right)$ for $x \in [0,1)$.

For the second statement, notice that $\Prx{R_j} = \Prx{T_j} \cdot \Prx{R_j |T_j} = \prod_{k=i}^{j-1}(1-p_k)) \cdot p_j r_j$. We rewrite the left hand side as the following.
\[
\eqref{eq:case1} + \eqref{eq:case3} = \prod_{k=1}^{n}(1-p_k) \cdot \left( \sum_{j=1}^{i-1} \frac{p_j}{1-p_j} + \sum_{j=i}^{n} \frac{p_j r_j}{1-p_j} \right)
\]
Consider the function $F: [0,1]^{n-i+1} \to [0,1]$:
\[
F(r_i,r_{i+1},...,r_{n}) \eqdef \sum_{j=i}^{n} \frac{p_j r_j}{1-p_j} - \ln \left( \frac{1}{1-\sum_{j=i}^{n}\prod_{k=i}^{j-1}(1-p_k) \cdot p_j r_j} \right)~.
\]
Notice that 
\begin{multline*}
\frac{\partial F}{\partial r_n} = \frac{p_n}{1-p_n} - \frac{\prod_{k=i}^{n-1}(1-p_k) \cdot p_n}{1-\sum_{j=i}^{n}\prod_{k=i}^{j-1}(1-p_k)\cdot p_j r_j} \\
\ge \frac{p_n}{1-p_n} - \frac{\prod_{k=i}^{n-1}(1-p_k) \cdot p_n}{1-\sum_{j=i}^{n}\prod_{k=i}^{j-1}(1-p_k)\cdot p_j} = \frac{p_n}{1-p_n} - \frac{\prod_{k=i}^{n-1}(1-p_k) \cdot p_n}{\prod_{k=i}^{n}(1-p_k)} = 0,
\end{multline*}
where we use the fact that $r_j \le 1$ for every $i\le j \le n$.
Consequently, $F(r_i,...,r_n) \ge F(r_i,...,r_{n-1},0)$. Similarly, one can prove that $F(r_i,...,r_k,0,...,0) \ge F(r_i,...,r_{k-1},0,0,...,0)$ for every $i \le k \le n$, and therefore, $F(r_i,...,r_n) \ge F(0,...,0) = 0$.
Together with the fact that $\frac{x}{1-x} \ge \ln\left(\frac{1}{1-x}\right)$ for all $x \in [0,1)$, we conclude the proof of the second statement:
\begin{multline*}
\eqref{eq:case1} + \eqref{eq:case3} \ge \prod_{k=1}^{n}(1-p_k) \cdot \left( \sum_{j=1}^{i-1} \ln \left(\frac{1}{1-p_j}\right) + \ln \left( \frac{1}{1-\sum_{j=i}^{n}\prod_{k=i}^{j-1}(1-p_k) \cdot p_j r_j} \right) \right)  \\
= \lambda \cdot \ln \left( \frac{1}{\prod_{j<i}(1-p_j) \cdot \left(1-\sum_{j=i}^{n}\prod_{k=i}^{j-1}(1-p_k) \cdot p_j r_j \right)}\right) = \lambda \cdot \ln \left( \frac{1}{\rho} \right)~.
\end{multline*}

Finally, for the third statement, we have
\begin{multline*}
\eqref{eq:case2} \le \prod_{k<i} (1-p_k) \cdot \sum_{j=i}^{n} \Prx{A_j}
= \prod_{k<i} (1-p_k) \cdot \left( \sum_{j=i}^{n} \Prx{T_j} - \sum_{j=i}^{n} \Prx{R_j} \right) \\
= \prod_{k<i} (1-p_k) \cdot \left( 1- \prod_{k=i}^{n}(1-p_k) - \sum_{j=i}^{n} \Prx{R_j}\right) = \rho - \lambda
\end{multline*}
\end{proof}

By rearranging \eqref{eq:opt} and by the definition of $\rho$, we have that
\[
\opt_i(\tau) \le \frac{\eqref{eq:case2}}{\rho} \le \frac{\alg - \eqref{eq:case1} - \eqref{eq:case3}}{\rho} \le \frac{\alg - \lambda \cdot \ln \left( \frac{1}{\rho}\right)}{\rho}~.
\]

\noindent \textbf{Case 1:} if $\ln \left( \frac{1}{\rho} \right) \ge \rho$, we have
    \begin{align*}
    \frac{\alg}{\opt} \ge \frac{\eqref{eq:case1} + \eqref{eq:case2} + \eqref{eq:case3}}{\lambda + \frac{\eqref{eq:case2}}{\rho}} \ge \frac{\lambda \cdot \ln\left( \frac{1}{\rho} \right) + \eqref{eq:case2}}{\lambda + \frac{\eqref{eq:case2}}{\rho}} \ge \frac{\lambda \cdot \ln\left( \frac{1}{\rho} \right) + \rho - \lambda}{\lambda + \frac{\rho - \lambda}{\rho}} \ge \Gamma^*, 
    \end{align*}
where the second inequality follows from the second statement of Lemma~\ref{lem:prob_alg}, and the third inequality follows from the third statement of Lemma~\ref{lem:prob_alg} and the assumption that $\ln \left( \frac{1}{\rho} \right) \ge \rho$.

\noindent \textbf{Case 2:} if $\ln \left( \frac{1}{\rho} \right) < \rho$, we have
    \begin{align*}
    \frac{\alg}{\opt} \ge \frac{\alg}{\lambda + \frac{\alg - \lambda \cdot \ln \left( \frac{1}{\rho}\right)}{\rho}} \ge \frac{\lambda \cdot \ln\left( \frac{1}{\lambda} \right)}{\lambda + \frac{\lambda}{\rho} \cdot \ln\left( \frac{\rho}{\lambda} \right)} \ge \frac{\lambda \cdot \ln\left( \frac{1}{\lambda} \right)}{\lambda + \frac{1}{e}} = \Gamma^*,
    \end{align*}
where the second inequality follows from the first statement of Lemma~\ref{lem:prob_alg} and the assumption that $\ln \left( \frac{1}{\rho} \right) < \rho$, and the third inequality follows from the fact that $x \ln \left(\frac{1}{x}\right) \le \frac{1}{e}$ for $x \in [0,1]$.
\end{proof} \subsection{Lower bound}

Next, we show that single-threshold algorithms are optimal among all deterministic identity-blind algorithms. Specifically, the identity-blindness gap of deterministic identity-blind algorithms is at most $\Gamma^*$.

As a warm up, we first prove that the guarantee of Theorem~\ref{thm:max_prob_single} is tight among all single-threshold algorithms. Within this section, we use $\lambda \approx 0.245$ to denote the solution of the following function and $\rho \approx 0.513$ to denote the corresponding minimizer of the left hand side of the equation.
\[
\min_{\rho \ge \lambda} \frac{\lambda \cdot \ln\left( \frac{1}{\rho} \right) + \rho - \lambda}{\lambda + 1 - \frac{\lambda}{\rho}} = \frac{\lambda \cdot \ln\left( \frac{1}{\lambda} \right)}{\lambda + \frac{1}{e}}
\]

\begin{theorem}
\label{thm:prob_single_neg}
    For the max-probability objective, no single-threshold algorithm achieves a better identity-blindness gap than $\Gamma^* \approx 0.562$.
\end{theorem}
\begin{proof}
Consider an instance that consists of $n$ boxes, where the $i$-th box has value $i$ with probability $\varepsilon$ and $0$ otherwise, where $\varepsilon \to 0$ and $\eps \gg \frac{1}{n}$. Consider an arbitrary single-threshold algorithm, and let $T$ be its threshold. Without loss of generality, we assume that $T$ is an integer and the single-threshold algorithm always accepts the box $T$ when $v_T = T$. We distinguish between two cases, depending on the value of $\lambda' \eqdef (1-\varepsilon)^{n-T+1}$.

\begin{figure}[!htb]
   \begin{minipage}{0.48\textwidth}
     \centering
     \includegraphics[width=.9\linewidth]{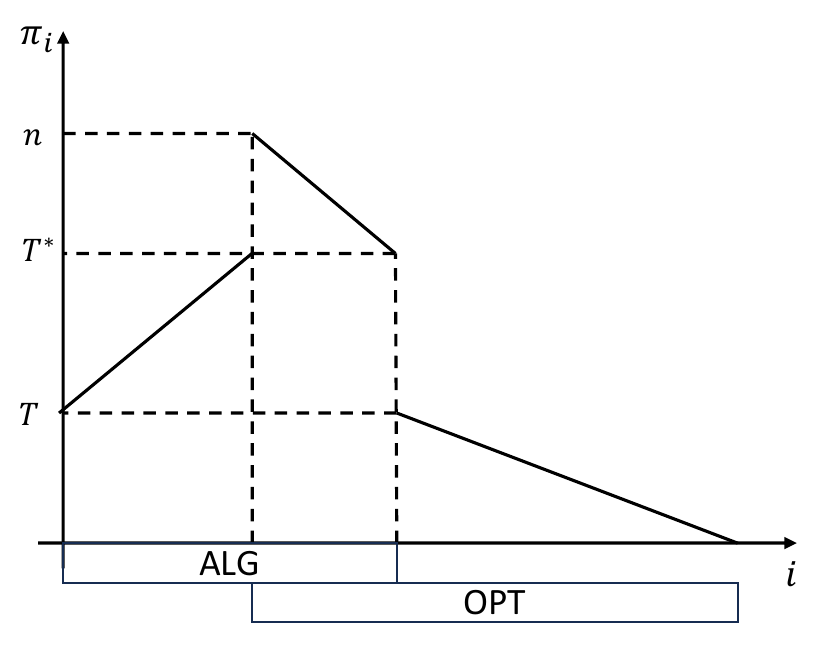}
     \caption{Case 1a.}\label{fig:case1a}
   \end{minipage}\hfill
   \begin{minipage}{0.48\textwidth}
     \centering
     \includegraphics[width=.9\linewidth]{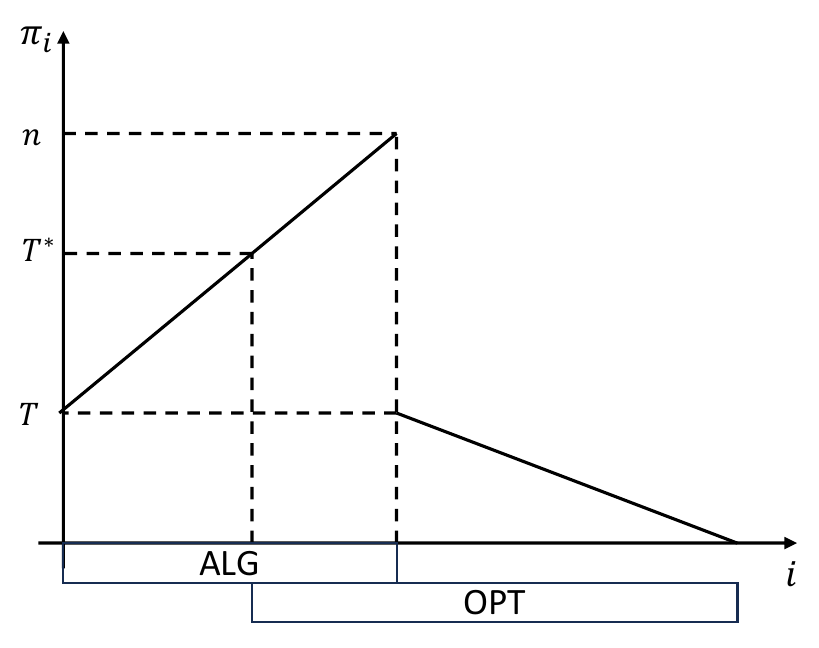}
     \caption{Case 2.}\label{fig:case2}
   \end{minipage}
\end{figure}

\paragraph{Case 1: $\lambda' \ge \lambda$.} Let $T^* = n+1 - \log_{1-\varepsilon} \frac{\lambda}{\rho}$. We further consider two cases.
\paragraph{Case 1a: $T^*>T$.} Consider the following arrival order (refer to figure~\ref{fig:case1a}), where
\[
\pi_i =
\begin{cases}
\text{Box } T+i-1 & \text{if } 1 \le i \le T^*-T \\
\text{Box } n+T^*-T+1 - i & \text{if } T^*-T+1 \le i \le n-T+1 \\
\text{Box } n+1-i & \text{if } n-T+1 < i \le n
\end{cases}
\]
Let $X_1$ (respectively, $X_2, X_3$) denote the random variable indicating the number of realized values below $T$ (respectively, between $T$ and $T^*$, and above $T^*$). Then, we have
\begin{multline*}
\alg = \Prx{X_2 = 1 \wedge X_3=0} + \Prx{X_2 = 0 \wedge X_3 >0} \\
= (T^* - T) \cdot \varepsilon \cdot (1-\varepsilon)^{n-T} + (1-\varepsilon)^{T^*-T} \cdot \left(1 - (1-\varepsilon)^{n-T^*+1} \right) \approx \lambda' \cdot \ln \left( \frac{\lambda}{\rho \lambda'} \right) + \frac{\rho \lambda'}{\lambda} - \lambda'~.
\end{multline*}
On the other hand, an adaptive algorithm can skip the first $T^*-T$ boxes and then accept the first realized box afterwards. Then, we have
\begin{multline*}
\opt \ge \Prx{X_3 > 0} + \Prx{X_1>0 \wedge X_2=0 \wedge X_3=0} \\
= \left(1 - (1-\varepsilon)^{n-T^*+1} \right) + \left(1-(1-\varepsilon)^{T-1} \right) \cdot (1-\varepsilon)^{n-T+1} \approx 1 - \frac{\lambda}{\rho} + \lambda' ~.
\end{multline*}
Therefore, the identity-blindness gap is at most 
\[
\frac{\lambda' \cdot \ln \left( \frac{\lambda}{\rho \lambda'} \right) + \frac{\rho \lambda'}{\lambda} - \lambda'}{1 - \frac{\lambda}{\rho} + \lambda'} \le \frac{\lambda \cdot \ln \left( \frac{1}{\rho} \right) + \rho - \lambda}{1 - \frac{\lambda}{\rho} + \lambda} = \Gamma^*,
\]
where we use the fact that the left hand side is a decreasing function for $\lambda' \in [\lambda,1]$.

\paragraph{Case 1b: $T^* \le T$.} Consider the arrival order when boxes arrive in descending order, i.e., $ \pi_i = n+1-i$. Then the optimal algorithm wins with probability $1$ while the single-threshold algorithm wins on if the maximum value is at least $T \ge T^*$. Consequently,
\[
\frac{\alg}{\opt} \le 1-(1-\varepsilon)^{n-T^*+1} = 1 - \frac{\lambda}{\rho} \approx 0.522 < \Gamma^*~.
\]

\paragraph{Case 2: $\lambda' < \lambda$.} Let $T^* = n+1 - \log_{1-\varepsilon} \frac{1}{e}$.
Consider the following arrival order (refer to figure~\ref{fig:case2}), where
\[
\pi_i =
\begin{cases}
\text{Box } T+i-1 & \text{if } 1 \le i \le n-T+1 \\
\text{Box } n+1-i & \text{if } n-T+1 < i \le n
\end{cases}
\]
It is straightforward to verify that the single-threshold algorithm wins if and only if exactly one value is realized above $T$, i.e.,
\[
\alg = (n-T+1) \cdot \varepsilon \cdot (1-\varepsilon)^{n-T} \approx \lambda' \cdot \ln \left( \frac{1}{\lambda'} \right)~.
\]
On the other hand, the optimal algorithm can skip the first $T^*-T$ boxes and then accept the first realized box. By doing so, the algorithm wins when exactly one box is realized with value above $T^*$ or no box is realized with value above $T$.
\[
\opt \ge (n-T^*+1) \cdot \varepsilon \cdot (1-\varepsilon)^{n-T^*+1} + (1-\varepsilon)^{n-T+1} \approx \frac{1}{e} + \lambda'~.
\]
Therefore, the identity-blindness gap is at most 
\[
\frac{\lambda' \cdot \ln\left( \frac{1}{\lambda'} \right)}{\lambda' + \frac{1}{e}} \le \frac{\lambda \cdot \ln\left( \frac{1}{\lambda} \right)}{\lambda + \frac{1}{e}} = \Gamma^*,
\]
where we use the fact that the left hand side is an increasing function for $\lambda' \in [0,\lambda]$.
\end{proof}

Finally, we prove that the single-threshold algorithm from Theorem~\ref{thm:max_prob_single} achieves the best possible identity-blindness gap among all deterministic algorithms. Our construction exploits the behavior of a fixed deterministic identity-blind algorithm.

\begin{theorem}
    For any deterministic identity-blind algorithm, the identity-blindness gap is at most $\Gamma^* \approx 0.562$. 
\end{theorem}
\begin{proof}
Let there be $n$ real boxes, where the $i$-th box has value $i$ with probability $\varepsilon$ and $0$ otherwise. These boxes play the same role as our construction in Theorem~\ref{thm:prob_single_neg}. Similarly, we assume that $n \to \infty, \varepsilon \to 0$, and $\varepsilon \gg \frac{1}{n}$.
In addition, let there be $2^n - n$ dummy boxes that have deterministic values of $0$.

Let $a = \left\lfloor \frac{1}{\varepsilon} \cdot \left(\ln\left(\frac{1}{\lambda}\right)-1\right) \right\rfloor$, $t_1 = \left\lfloor \frac{1}{\varepsilon} \right\rfloor$, and $t_2 = \left\lfloor \frac{1}{\varepsilon} \cdot \ln\left(\frac{\rho}{\lambda}\right) \right\rfloor$.
Let $\alg$ be an arbitrary deterministic identity-blind algorithm. We construct an instance according to the behavior of $\alg$. We only specify the arrival order of the real boxes. Below, $i_k$ denote the arrival time of the $k$-th real box, i.e., $\pi(i_k) = k$.
\begin{tcolorbox}
Initialize $\ell = 0, u = 2^n$; \hfill lower/upper indices \\
Initialize $A = R = \varnothing$. \hfill accept/reject sets
\paragraph{First Stage.} For $k$ from $1$ to $n-t_1$: Let $i_k = \frac{\ell+u}{2}$. Consider the behavior of $\alg$ when the first $i_k-1$ boxes have value $0$ and the $i_k$-th box arrives with value $k$. 
\begin{itemize}
    \item If $\alg$ accepts it, let $\ell = i_k$, $A = A \cup \{k\}$.
    \item Else, let $u = i_k$, $R = R \cup \{k\}$.
\end{itemize}
\paragraph{Second Stage.}
\paragraph{If $|A| < a$:}
for $k$ from $\left(n-t_1+1\right)$ to $(n - t_2)$: Let $i_k = \frac{\ell+u}{2}$. Consider the behavior of $\alg$ when the first $i_k-1$ boxes have value $0$ and the $i_k$-th box arrives with value $k$.
\begin{itemize}
    \item If $\alg$ accepts it, let $\ell = i_k$, $A = A \cup \{k\}$.
    \item Else, let $u = i_k$, $R = R \cup \{k\}$.
\end{itemize}
Finally, let the remaining boxes $n,n-1,...,n-t_2+1$ arrive sequentially in descending order from $i_{n-t_2}+1$ to $i_{n-t_2} + t_2$. We use $B$ to denote these boxes.
\paragraph{Else $|A| \ge a$:}
let the remaining boxes $n-t_1+1,...,n$ arrive sequentially in ascending order from $i_{n-t_1}+1$ to $i_{n-t_1} + t_1$. We use $B$ to denote these boxes.
\end{tcolorbox}

In both cases, the real boxes come in the order of $A \to B \to R$, where boxes in $A$ come in ascending order and boxes in $R$ come in descending order. Moreover, for an arbitrary box $k \in A$, $\alg$ would like to accept it if it is the first realized box; for an arbitrary box $k \in R$, $\alg$ would like to reject it if it is the first realized box.

Consequently, if the first realized box belongs to $A$, $\alg$ would accept it and wins only when all other boxes of $A \cup B$ are not realized; if the first realized box belongs to $R$, $\alg$ would reject it and lose the game since all remaining boxes have smaller values.

We have not specified the behavior of $\alg$ if the first realized box belongs to $B$. It is straightforward to verify that the best strategy is to accept it for both cases. We omit the tedious case analysis.

Finally, we calculate the identity-blindness gap of $\alg$ with respect to the two cases depending on our construction of the second stage. The first case is similar to the case 1a of Theorem~\ref{thm:prob_single_neg}. Note that $|A|<a$ at the end of the first stage, it must be the case that $|A| < a+t_1-t_2$ at the end of the second stage. Let $x \eqdef (1-\varepsilon)^{|A|} > (1-\varepsilon)^{a+t_1-t_2} \approx \rho$
\[
\alg \le |A| \cdot \varepsilon \cdot (1- \varepsilon)^{|A|+|B|-1} + ( 1-\varepsilon)^{|A|} \cdot \left( 1 - \left( 1 - \varepsilon \right)^{|B|} \right) \approx \frac{x\lambda}{\rho} \cdot \ln \left(\frac{1}{x}\right) + x - \frac{x \lambda}{\rho},
\]
where the first term corresponds to the case when the first realized box belongs to $A$ and the second term corresponds to the case when the first realized box belongs to $B$. On the other hand, an adaptive algorithm that knows the order could accept the first realized box in $B \cup R$. It wins when at least $1$ box is realized in $B$, or no box is realized in $A \cup B$. Therefore,
\[
\opt \ge \left( 1 - \left( 1 - \varepsilon \right)^{|B|} \right) + \left( 1 - \varepsilon \right)^{|A|+|B|} \approx 1 - \frac{\lambda}{\rho} + \frac{x \lambda}{\rho}~.
\]
Together, we have
\[
\frac{\alg}{\opt} \lesssim \frac{\frac{x\lambda}{\rho} \cdot \ln \left(\frac{1}{x}\right) + x - \frac{x \lambda}{\rho}}{1 - \frac{\lambda}{\rho} + \frac{x \lambda}{\rho}} \le \frac{\lambda \cdot \ln \left(\frac{1}{\rho}\right) + \rho - \lambda}{1 - \frac{\lambda}{\rho} + \lambda} = \Gamma^*,
\]
where the second inequality holds since the function is decreasing for $x \in [\rho,1]$. 

The second case is similar to the case 2 of Theorem~\ref{thm:prob_single_neg}. Let \[
y \eqdef (1-\varepsilon)^{|A|+|B|} \le (1-\varepsilon)^{a+t_1} \approx \lambda~.
\]
Then, we have
\[
\alg \le \left(|A| + |B| \right) \cdot \varepsilon \cdot \left( 1-\varepsilon \right)^{|A|+|B|-1} \approx y \cdot \ln \left( \frac{1}{y}\right)~,
\]
where $\alg$ wins only if exactly one box in $A \cup B$ is realized.
On the other hand, an adaptive algorithm that knows the order could accept the first realized box in $B \cup R$. It wins when exactly $1$ box is realized in $B$, or no box is realized in $A \cup B$. Therefore,
\[
\opt \ge |B| \cdot \varepsilon \cdot ( 1-\varepsilon)^{|B|-1} + \left( 1-\varepsilon \right)^{|A|+|B|} \approx \frac{1}{e} + y~,
\]
where for the approximaiton we use the definition of $y$, and since in the second case, $|B|=t_1 \approx \frac{1}{\varepsilon}$.
Together, we have 
\[
\frac{\alg}{\opt} \lesssim \frac{y \ln \left( \frac{1}{y}\right)}{y + \frac{1}{e}} \le \frac{\lambda \ln \left( \frac{1}{\lambda}\right)}{\lambda+ \frac{1}{e}} = \Gamma^*,
\]
where the second inequality holds as the function is increasing in $y \in [0, \lambda]$.
This concludes the proof of the theorem.
\end{proof}

\bibliographystyle{abbrvnat}
\bibliography{order}

\begin{thebibliography}{39}
\providecommand{\natexlab}[1]{#1}
\providecommand{\url}[1]{\texttt{#1}}
\expandafter\ifx\csname urlstyle\endcsname\relax
  \providecommand{\doi}[1]{doi: #1}\else
  \providecommand{\doi}{doi: \begingroup \urlstyle{rm}\Url}\fi

\bibitem[Agrawal et~al.(2020)Agrawal, Sethuraman, and Zhang]{AgrawalSZ20}
S.~Agrawal, J.~Sethuraman, and X.~Zhang.
\newblock On optimal ordering in the optimal stopping problem.
\newblock In P.~Bir{\'{o}}, J.~D. Hartline, M.~Ostrovsky, and A.~D. Procaccia,
  editors, \emph{{EC} '20: The 21st {ACM} Conference on Economics and
  Computation, Virtual Event, Hungary, July 13-17, 2020}, pages 187--188.
  {ACM}, 2020.

\bibitem[Arsenis and Kleinberg(2022)]{10.1145/3490486.3538301}
M.~Arsenis and R.~Kleinberg.
\newblock Individual fairness in prophet inequalities.
\newblock In \emph{Proceedings of the 23rd ACM Conference on Economics and
  Computation}, EC '22, page 245, New York, NY, USA, 2022. Association for
  Computing Machinery.
\newblock ISBN 9781450391504.
\newblock \doi{10.1145/3490486.3538301}.
\newblock URL \url{https://doi.org/10.1145/3490486.3538301}.

\bibitem[Arsenis et~al.(2021)Arsenis, Drosis, and
  Kleinberg]{arsenis2021constrained}
M.~Arsenis, O.~Drosis, and R.~Kleinberg.
\newblock Constrained-order prophet inequalities.
\newblock In \emph{Proceedings of the 2021 ACM-SIAM Symposium on Discrete
  Algorithms (SODA)}, pages 2034--2046. SIAM, 2021.

\bibitem[Azar et~al.(2014)Azar, Kleinberg, and Weinberg]{azar2014prophet}
P.~D. Azar, R.~Kleinberg, and S.~M. Weinberg.
\newblock Prophet inequalities with limited information.
\newblock In \emph{Proceedings of the twenty-fifth annual ACM-SIAM symposium on
  Discrete algorithms}, pages 1358--1377. SIAM, 2014.

\bibitem[Azar et~al.(2018)Azar, Chiplunkar, and Kaplan]{azar2018prophet}
Y.~Azar, A.~Chiplunkar, and H.~Kaplan.
\newblock Prophet secretary: Surpassing the 1-1/e barrier.
\newblock In \emph{Proceedings of the 2018 ACM Conference on Economics and
  Computation}, pages 303--318, 2018.

\bibitem[Beyhaghi et~al.(2018)Beyhaghi, Golrezaei, Leme, Pal, and
  Sivan]{beyhaghi2018improved}
H.~Beyhaghi, N.~Golrezaei, R.~P. Leme, M.~Pal, and B.~Sivan.
\newblock Improved approximations for free-order prophets and second-price
  auctions.
\newblock \emph{arXiv preprint arXiv:1807.03435}, 2018.

\bibitem[Braverman et~al.(2022)Braverman, Derakhshan, and
  Lovett]{braverman2022max}
M.~Braverman, M.~Derakhshan, and A.~M. Lovett.
\newblock Max-weight online stochastic matching: Improved approximations
  against the online benchmark.
\newblock In \emph{{EC} '22: The 23rd {ACM} Conference on Economics and
  Computation, Boulder, CO, USA, July 11 - 15, 2022}, pages 967--985. {ACM},
  2022.
\newblock \doi{10.1145/3490486.3538315}.
\newblock URL \url{https://doi.org/10.1145/3490486.3538315}.

\bibitem[Bubna and Chiplunkar(2023)]{BubnaC23}
A.~Bubna and A.~Chiplunkar.
\newblock Prophet inequality: Order selection beats random order.
\newblock \emph{EC}, 2023.

\bibitem[Chawla et~al.(2010)Chawla, Hartline, Malec, and
  Sivan]{DBLP:conf/stoc/ChawlaHMS10}
S.~Chawla, J.~D. Hartline, D.~L. Malec, and B.~Sivan.
\newblock Multi-parameter mechanism design and sequential posted pricing.
\newblock In L.~J. Schulman, editor, \emph{Proceedings of the 42nd {ACM}
  Symposium on Theory of Computing, {STOC} 2010, Cambridge, Massachusetts, USA,
  5-8 June 2010}, pages 311--320. {ACM}, 2010.
\newblock \doi{10.1145/1806689.1806733}.
\newblock URL \url{https://doi.org/10.1145/1806689.1806733}.

\bibitem[Correa et~al.(2021{\natexlab{a}})Correa, Cristi, Duetting, and
  Fard]{50313}
J.~Correa, A.~Cristi, P.~Duetting, and A.~N. Fard.
\newblock Fairness and bias in online selection.
\newblock In \emph{Proceedings of the 2021 International Conference on Machine
  Learning (ICML'21)}, pages 2112--2121, 2021{\natexlab{a}}.

\bibitem[Correa et~al.(2021{\natexlab{b}})Correa, Saona, and
  Ziliotto]{CorreaSZ21}
J.~R. Correa, R.~Saona, and B.~Ziliotto.
\newblock Prophet secretary through blind strategies.
\newblock \emph{Math. Program.}, 190\penalty0 (1):\penalty0 483--521,
  2021{\natexlab{b}}.

\bibitem[D{\"u}tting and Kleinberg(2015)]{dutting2015polymatroid}
P.~D{\"u}tting and R.~Kleinberg.
\newblock Polymatroid prophet inequalities.
\newblock In \emph{Algorithms-ESA 2015}, pages 437--449. Springer, 2015.

\bibitem[D{\"u}tting et~al.(2020)D{\"u}tting, Kesselheim, and
  Lucier]{dutting2020log}
P.~D{\"u}tting, T.~Kesselheim, and B.~Lucier.
\newblock An o (log log m) prophet inequality for subadditive combinatorial
  auctions.
\newblock \emph{ACM SIGecom Exchanges}, 18\penalty0 (2):\penalty0 32--37, 2020.

\bibitem[D{\"{u}}tting et~al.(2023)D{\"{u}}tting, Gergatsouli, Rezvan, Teng,
  and Tsigonias{-}Dimitriadis]{DBLP:conf/sigecom/DuttingGRTT23}
P.~D{\"{u}}tting, E.~Gergatsouli, R.~Rezvan, Y.~Teng, and
  A.~Tsigonias{-}Dimitriadis.
\newblock Prophet secretary against the online optimal.
\newblock In K.~Leyton{-}Brown, J.~D. Hartline, and L.~Samuelson, editors,
  \emph{Proceedings of the 24th {ACM} Conference on Economics and Computation,
  {EC} 2023, London, United Kingdom, July 9-12, 2023}, pages 561--581. {ACM},
  2023.
\newblock \doi{10.1145/3580507.3597736}.
\newblock URL \url{https://doi.org/10.1145/3580507.3597736}.

\bibitem[Ehsani et~al.(2018)Ehsani, Hajiaghayi, Kesselheim, and
  Singla]{ehsani2018prophet}
S.~Ehsani, M.~Hajiaghayi, T.~Kesselheim, and S.~Singla.
\newblock Prophet secretary for combinatorial auctions and matroids.
\newblock In \emph{Proceedings of the twenty-ninth annual acm-siam symposium on
  discrete algorithms}, pages 700--714. SIAM, 2018.

\bibitem[Esfandiari et~al.(2017)Esfandiari, Hajiaghayi, Liaghat, and
  Monemizadeh]{esfandiari2017prophet}
H.~Esfandiari, M.~Hajiaghayi, V.~Liaghat, and M.~Monemizadeh.
\newblock Prophet secretary.
\newblock \emph{SIAM Journal on Discrete Mathematics}, 31\penalty0
  (3):\penalty0 1685--1701, 2017.

\bibitem[Esfandiari et~al.(2020)Esfandiari, Hajiaghayi, Lucier, and
  Mitzenmacher]{esfandiariHLM20}
H.~Esfandiari, M.~Hajiaghayi, B.~Lucier, and M.~Mitzenmacher.
\newblock Prophets, secretaries, and maximizing the probability of choosing the
  best.
\newblock In S.~Chiappa and R.~Calandra, editors, \emph{The 23rd International
  Conference on Artificial Intelligence and Statistics, {AISTATS} 2020, 26-28
  August 2020, Online [Palermo, Sicily, Italy]}, volume 108 of
  \emph{Proceedings of Machine Learning Research}, pages 3717--3727. {PMLR},
  2020.
\newblock URL \url{http://proceedings.mlr.press/v108/esfandiari20a.html}.

\bibitem[Ezra and Garbuz(2023)]{DBLP:conf/wine/EzraG23}
T.~Ezra and T.~Garbuz.
\newblock The importance of knowing the arrival order in combinatorial bayesian
  settings.
\newblock In J.~Garg, M.~Klimm, and Y.~Kong, editors, \emph{Web and Internet
  Economics - 19th International Conference, {WINE} 2023, Shanghai, China,
  December 4-8, 2023, Proceedings}, volume 14413 of \emph{Lecture Notes in
  Computer Science}, pages 256--271. Springer, 2023.
\newblock \doi{10.1007/978-3-031-48974-7\_15}.
\newblock URL \url{https://doi.org/10.1007/978-3-031-48974-7\_15}.

\bibitem[Ezra et~al.(2018)Ezra, Feldman, and Nehama]{ezra2018prophets}
T.~Ezra, M.~Feldman, and I.~Nehama.
\newblock Prophets and secretaries with overbooking.
\newblock In \emph{Proceedings of the 2018 {ACM} Conference on Economics and
  Computation, {EC}}, pages 319--320. {ACM}, 2018.

\bibitem[Ezra et~al.(2020)Ezra, Feldman, Gravin, and Tang]{EzraFGT20}
T.~Ezra, M.~Feldman, N.~Gravin, and Z.~G. Tang.
\newblock Online stochastic max-weight matching: Prophet inequality for vertex
  and edge arrival models.
\newblock In \emph{{EC}}, pages 769--787. {ACM}, 2020.

\bibitem[Ezra et~al.(2023)Ezra, Feldman, Gravin, and Tang]{ezra23order}
T.~Ezra, M.~Feldman, N.~Gravin, and Z.~G. Tang.
\newblock “who is next in line?” on the significance of knowing the arrival
  order in bayesian online settings.
\newblock In \emph{Proceedings of the 2023 Annual ACM-SIAM Symposium on
  Discrete Algorithms (SODA)}, pages 3759--3776. Society for Industrial and
  Applied Mathematics, 2023.

\bibitem[Feldman et~al.(2015)Feldman, Gravin, and Lucier]{FeldmanGL15}
M.~Feldman, N.~Gravin, and B.~Lucier.
\newblock Combinatorial auctions via posted prices.
\newblock In P.~Indyk, editor, \emph{Proceedings of the Twenty-Sixth Annual
  {ACM-SIAM} Symposium on Discrete Algorithms, {SODA} 2015, San Diego, CA, USA,
  January 4-6, 2015}, pages 123--135. {SIAM}, 2015.

\bibitem[Goldin and Rouse(2000)]{ClaudiaC00}
C.~Goldin and C.~Rouse.
\newblock Orchestrating impartiality: The impact of "blind" auditions on female
  musicians.
\newblock \emph{American Economic Review}, 90\penalty0 (4):\penalty0 715--741,
  September 2000.
\newblock \doi{10.1257/aer.90.4.715}.
\newblock URL \url{https://www.aeaweb.org/articles?id=10.1257/aer.90.4.715}.

\bibitem[Gravin and Wang(2019)]{GravinW19}
N.~Gravin and H.~Wang.
\newblock Prophet inequality for bipartite matching: Merits of being simple and
  non adaptive.
\newblock In \emph{{EC}}, pages 93--109. {ACM}, 2019.

\bibitem[Hajiaghayi et~al.(2007)Hajiaghayi, Kleinberg, and
  Sandholm]{hajiaghayi2007automated}
M.~T. Hajiaghayi, R.~Kleinberg, and T.~Sandholm.
\newblock Automated online mechanism design and prophet inequalities.
\newblock In \emph{AAAI}, volume~7, pages 58--65, 2007.

\bibitem[Kennedy(1985)]{kennedy1985optimal}
D.~P. Kennedy.
\newblock Optimal stopping of independent random variables and maximizing
  prophets.
\newblock \emph{The Annals of Probability}, pages 566--571, 1985.

\bibitem[Kennedy(1987)]{kennedy1987prophet}
D.~P. Kennedy.
\newblock Prophet-type inequalities for multi-choice optimal stopping.
\newblock \emph{Stochastic Processes and their applications}, 24\penalty0
  (1):\penalty0 77--88, 1987.

\bibitem[Kertz(1986)]{kertz1986comparison}
R.~P. Kertz.
\newblock Comparison of optimal value and constrained maxima expectations for
  independent random variables.
\newblock \emph{Advances in applied probability}, 18\penalty0 (2):\penalty0
  311--340, 1986.

\bibitem[Kessel et~al.(2021)Kessel, Saberi, Shameli, and
  Wajc]{kessel2021stationary}
K.~Kessel, A.~Saberi, A.~Shameli, and D.~Wajc.
\newblock The stationary prophet inequality problem.
\newblock \emph{arXiv preprint arXiv:2107.10516}, 2021.

\bibitem[Kleinberg and Weinberg(2012)]{kleinberg2012matroid}
R.~Kleinberg and S.~M. Weinberg.
\newblock Matroid prophet inequalities.
\newblock In \emph{Proceedings of the forty-fourth annual ACM symposium on
  Theory of computing}, pages 123--136, 2012.

\bibitem[Krengel and Sucheston(1977)]{krengel1977semiamarts}
U.~Krengel and L.~Sucheston.
\newblock Semiamarts and finite values.
\newblock \emph{Bulletin of the American Mathematical Society}, 83\penalty0
  (4):\penalty0 745--747, 1977.

\bibitem[Krengel and Sucheston(1978)]{krengel1978semiamarts}
U.~Krengel and L.~Sucheston.
\newblock On semiamarts, amarts, and processes with finite value.
\newblock \emph{Probability on Banach spaces}, 4:\penalty0 197--266, 1978.

\bibitem[Naor et~al.(2023)Naor, Srinivasan, and Wajc]{srinivasan2023online}
J.~Naor, A.~Srinivasan, and D.~Wajc.
\newblock Online dependent rounding schemes.
\newblock \emph{CoRR}, abs/2301.08680, 2023.

\bibitem[Niazadeh et~al.(2018)Niazadeh, Saberi, and
  Shameli]{niazadeh2018prophet}
R.~Niazadeh, A.~Saberi, and A.~Shameli.
\newblock Prophet inequalities vs. approximating optimum online.
\newblock In \emph{Web and Internet Economics: 14th International Conference,
  WINE 2018, Oxford, UK, December 15--17, 2018, Proceedings 14}, pages
  356--374. Springer, 2018.

\bibitem[Papadimitriou et~al.(2021)Papadimitriou, Pollner, Saberi, and
  Wajc]{papadimitriou2021online}
C.~Papadimitriou, T.~Pollner, A.~Saberi, and D.~Wajc.
\newblock Online stochastic max-weight bipartite matching: Beyond prophet
  inequalities.
\newblock In \emph{Proceedings of the 22nd ACM Conference on Economics and
  Computation}, pages 763--764, 2021.

\bibitem[Peng and Tang(2022)]{PengT22}
B.~Peng and Z.~G. Tang.
\newblock Order selection prophet inequality: From threshold optimization to
  arrival time design.
\newblock In \emph{{FOCS}}, pages 171--178. {IEEE}, 2022.

\bibitem[Rubinstein(2016)]{rubinstein2016beyond}
A.~Rubinstein.
\newblock Beyond matroids: Secretary problem and prophet inequality with
  general constraints.
\newblock In \emph{Proceedings of the forty-eighth annual ACM symposium on
  Theory of Computing}, pages 324--332, 2016.

\bibitem[Saberi and Wajc(2021)]{SaberiW21}
A.~Saberi and D.~Wajc.
\newblock The greedy algorithm is not optimal for on-line edge coloring.
\newblock In \emph{{ICALP}}, volume 198 of \emph{LIPIcs}, pages 109:1--109:18.
  Schloss Dagstuhl - Leibniz-Zentrum f{\"{u}}r Informatik, 2021.

\bibitem[Samuel-Cahn(1984)]{samuel1984comparison}
E.~Samuel-Cahn.
\newblock Comparison of threshold stop rules and maximum for independent
  nonnegative random variables.
\newblock \emph{the Annals of Probability}, pages 1213--1216, 1984.

\end{thebibliography}

\appendix

\section{No-large-point-mass Assumption: Discussion}

\label{sec:discussion}

Our analysis of the max probability objective is tight with respect to instances that satisfy the assumption that for the distribution of the maximum, there are no large point masses.

We first mention that our positive result (Theorem~\ref{thm:max_prob_single}) can be adjusted to every instance, by using randomization in the following way:
Instead of calculating $\tau$ such that $$ \Prx{\max_{i\in[n]} v_i < \tau } = \lambda,$$
which does not always exists, one can
find a  value of $\tau$ such that $$ \Prx{\max_{i\in[n]} v_i < \tau } \leq  \lambda \leq \Prx{\max_{i\in[n]} v_i \leq \tau} ,$$ which always exists\footnote{There might be more than one value that satisfies this. If so, one can choose an arbitrary one.}. 
Then, if we denote by $p_i \eqdef \Prx{v_i < \tau}$, and $q_i \eqdef \Prx{v_i =\tau}$, then we can observe that either $\prod_{i \in [n]} p_i  = \prod_{i \in [n]} (p_i+q_i)  $, which implies that  $q_i=0$ for all $i$  (i.e., no distribution has a point mass in $\tau$), and the analysis of Theorem~\ref{thm:max_prob_single} works without changes;
otherwise, if we denote by $f(\xi) = \prod_{i \in [n]} (p_i+\xi \cdot q_i)$, then we know that 
$$f(0) =  \prod_{i \in [n]} p_i < \prod_{i \in [n]} (p_i+q_i) =f(1), $$ and $f$ is a strictly monotone function in $[0,1]$ which implies that there exists a unique value of $\xi$ such that  
$  f(\xi) = \lambda $.
The same analysis of Theorem~\ref{thm:max_prob_single} applies to the \textbf{randomized} algorithm that selects values that are strictly larger than $\tau$, and selects $\tau$  with probability $\xi$.
Overall, this argument shows:
\begin{claim}
    The above single-threshold randomized algorithm has an identity-blindness gap of $\Gamma^* \approx 0.562$
    for every instance.
\end{claim}

We next show that if we remove the assumption that the distribution of the maximum value has no large mass points, then an identity-blindness gap of $\Gamma^* \approx 0.562$ is not achievable for deterministic algorithms.

We first show an upper bound for deterministic algorithms.
\begin{claim}
    No identity-blind deterministic algorithm can guarantee an identity-blindness gap of more than $\frac{\mu}{1-\mu} \approx 0.517$, where $\mu \approx 0.341$ is the unique solution in $[0,1)$ to the following equation.
\[
\frac{\mu}{1-\mu} = \frac{\ln \left( \frac{1}{\mu}\right)}{\ln \left( \frac{1}{\mu}\right) +1}~.
\] 
Moreover, this can be shown even when the maximum value is unique with probability $1$.
\end{claim}
\begin{proof}
Let there be $n \to \infty$ real boxes, where the $i$-th box has value $i$ with probability $\varepsilon$ and $0$ otherwise. Here, $\varepsilon$ is the constant that $(1-\varepsilon)^n = \mu$,
Moreover, there is a special deterministic box with value $0.5$, and $n$ dummy boxes with deterministic values $0$.
Fix an arbitrary identity-blind deterministic algorithm $\alg$. Consider the following two cases depending on the behavior of $\alg$ when the first $n$ boxes have value $0$ and the $(n+1)$-th box has value $0.5$.
\paragraph{Case 1: Accept.}
Consider the instance when the $n$ dummy boxes arrived first, followed by the deterministic $0.5$ box, and finally the $n$ real boxes arrive in descending order.
Then, $\alg$ wins if and only if the $n$ real boxes are not realized, while the optimal algorithm should reject $0.5$ and accepts the first realized box (if exists) afterwards. Then the identity-blindness gap is at most $\frac{\mu}{1-\mu}$.
\paragraph{Case 2: Reject.}
Consider the instance when the $n$ real boxes arrived in ascending order first, followed by the deterministic $0.5$ box, followed by the $n$ dummy boxes.
It is straightforward to verify that the best strategy of $\alg$ is to accept greedily, i.e., accept the first realized box. In this case, it wins if and only if there is exactly one real box being realized. Consequently, $\alg = n \cdot \varepsilon \cdot (1-\varepsilon)^{n-1} \approx \mu \cdot \ln \left(\frac{1}{\mu}\right)$.
The optimal algorithm would do the same for the first $n$ boxes. Moreover, it would accept the deterministic $0.5$ box if none of the first $n$ boxes are realized. I.e., $\opt = n \cdot \varepsilon \cdot (1-\epsilon)^{n-1} + (1-\varepsilon) \approx \mu \cdot \ln \left(\frac{1}{\mu}\right) + \mu$. Thus, the identity-blindness gap is at most $\frac{\ln \left(\frac{1}{\mu}\right)}{\ln \left(\frac{1}{\mu}\right)+1}$.
\end{proof}

We next show an upper bound for deterministic single-threshold algorithms.
\begin{claim}
    No deterministic single-threshold algorithm (thus, identity-blind) can guarantee an identity-blindness gap of more than $\mu \approx0.4464$, where $\mu$ is the solution to the following equation. $$ \mu = \frac{\ln(\frac{1}{\mu})}{1+\ln(\frac{1}{\mu})}$$ 
\end{claim}
\begin{proof}
Let there be $n \to \infty$ real boxes, where the $i$-th box has value $i$ with probability $\varepsilon$ and $0$ otherwise. Here, $\varepsilon$ is the constant that $(1-\varepsilon)^n = \mu$,
Moreover, there are two special deterministic boxes with a value of $0.5$.
Now consider two arrival orders: in both of them, the first and last boxes are the deterministic boxes, and in one order, the boxes arrive in an increasing order, and in the other, in a decreasing order.
If the algorithm selects the deterministic box it selects the maximum with probability $\mu$, while for the order where the boxes arrive in a decreasing order, the optimal algorithm that rejects the first box, and selects the first non-zero box afterward, selects the maximum with probability $1$.
If the algorithm does not select the deterministic box, then for the increasing arrival order, it selects the maximum value with a probability of $ n \cdot \varepsilon \cdot (1-\varepsilon)^{n-1} \approx \mu \cdot \ln \left(\frac{1}{\mu}\right)$, while the optimal algorithm that rejects the first box, and selects the first non-zero box selects the maximum with probability  of $ n \cdot \varepsilon \cdot (1-\varepsilon)^{n-1} \approx \mu \cdot \ln \left(\frac{1}{\mu}\right) + \mu $.
Thus, the identity-blindness gap of single threshold algorithms is at most $\min\left\{\mu, \frac{\mu\ln (\frac{1}{\mu})}{\mu + \mu \ln (\frac{1}{\mu})}\right\}$.
\end{proof} 
\end{document}